\documentclass[acmtocl]{acmtrans2m}
\usepackage{amssymb,amsmath,proof,stmaryrd,txfonts}

\newtheorem{theorem}{Theorem}[section]
\newtheorem{corollary}[theorem]{Corollary}
\newtheorem{proposition}[theorem]{Proposition}
\newtheorem{lemma}[theorem]{Lemma}
\newdef{example}[theorem]{Example}
\newdef{definition}[theorem]{Definition}
\newtheorem{property}[theorem]{Property}

\newcommand{\de}{\dot=}

\newcommand\state[1]{\langle #1 \rangle}
\newcommand\chreq{\equiv}
\newcommand\ent{\rhd}
\newcommand\entv{\blacktriangleright}
\renewcommand\sim{{\blacktriangleleft\blacktriangleright}}
\newcommand\equv{\equiv_\vee}
\newcommand\equivll{\dashv\vdash}
\newcommand\subxy{\left[x/y\right]}
\newcommand\subxt{\left[x/t\right]}
\newcommand\subxa{\left[x/a\right]}
\newcommand\psubxt{\left[x\wr t\right] }
\newcommand\xet{\ensuremath{x\doteq t}}

\newcommand{\oesqv}{\ensuremath{\omega^\vee_e}}

\newcommand{\bbP}{\ensuremath{\mathbb{P}}}
\newcommand{\bbS}{\ensuremath{\mathbb{S}}}
\newcommand{\B}{\ensuremath{\mathbb{B}}}
\newcommand{\U}{\ensuremath{\mathbb{U}}}
\newcommand{\C}{\ensuremath{\mathbb{C}}}

\newcommand{\V}{\ensuremath{\mathbb{V}}}
\newcommand{\G}{\ensuremath{\mathbb{G}}}

\newcommand{\cA}{\ensuremath{\mathcal{A}}}
\newcommand{\cS}{\ensuremath{\mathcal{S}}}
\newcommand{\cC}{\ensuremath{\mathcal{C}}}
\newcommand{\bcA}{\ensuremath{\bar{\mathcal{A}}}}
\newcommand{\bcS}{\ensuremath{\bar{\mathcal{S}}}}
\newcommand{\bcC}{\ensuremath{\bar{\mathcal{C}}}}
\newcommand{\bcL}{\ensuremath{\bar{\mathcal{L}}}}

\newcommand{\cD}{\ensuremath{\mathcal{D}}}

\newcommand{\cL}{\mathcal{L}}

\newcommand{\cb}{\ensuremath{c_b(\bar t)}}
\newcommand{\cu}{\ensuremath{c_u(\bar t)}}

\newcommand{\Sct}{\ensuremath{\Sigma_{CT}}}
\newcommand{\Seq}{\ensuremath{\Sigma_{\doteq}}}
\newcommand{\Sp}{\ensuremath{\Sigma_\mathbb{P}}}

\newcommand{\hT}{\hat{T}}

\newcommand{\bS}{\bar{S}}
\newcommand{\bT}{\bar{T}}
\newcommand{\bU}{\bar{U}}
\newcommand{\bV}{\bar{V}}

\newcommand{\bt}{\bar{t}}
\renewcommand{\bt}{\bar{t}}
\newcommand{\by}{\bar{y}}
\newcommand{\bl}{\bar{l}}
\newcommand{\bs}{\bar{s}}
\newcommand{\bx}{\bar{x}}

\newcommand{\chrv}{CHR$^\vee$}

\newcommand{\oesq}{\ensuremath{\omega_{e}}}

\newcommand{\der}{\ensuremath{\mapsto}}
\newcommand{\nder}{\ensuremath{\not\mapsto}}

\newcommand\inintv[2]{\ensuremath{#1\in\{1,\ldots,#2\}}}

\newcommand{\x}{{\;\otimes\;}}
\newcommand{\lp}{\multimap}
\newcommand{\lpl}{\multimapboth}

\newcommand{\bang}{\; !}
\newcommand{\with}{\,\&\,}
\newcommand{\lone}{\boldsymbol{1}}
\newcommand{\lzero}{\boldsymbol{0}}
\newcommand{\ltop}{\boldsymbol{\top}}

\title{Linear-Logic Based Analysis of Constraint Handling Rules with Disjunction}
\author{HARIOLF BETZ and THOM FR\"UHWIRTH\\
University of Ulm\\
}

\begin{abstract}
Constraint Handling Rules (CHR) is a declarative committed-choice programming
language with a strong relationship to linear logic. Its generalization CHR with
Disjunction (CHR$^\vee$) is a multi-paradigm declarative programming language
that allows the embedding of horn programs.

We analyse the assets and the limitations of the classical declarative semantics
of CHR before we motivate and develop a linear-logic declarative semantics for
CHR and CHR$^\vee$.

We show how to apply the linear-logic semantics to decide program properties and
to prove operational equivalence of CHR$^\vee$ programs across the boundaries of
language paradigms.
\end{abstract}

\category{F.3.1}{Theory of Computation}{Logics and Meanings of Programs}[Specifying and Verifying and Reasoning about Programs]
\category{F.3.2}{Theory of Computation}{Logics and Meanings of Programs}[Semantics of Programming Languages]

\terms{Languages, Theory, Verification}

\keywords{Constraint Handling Rules, Linear Logic, Declarative Semantics}

\begin{document}

\begin{bottomstuff}
\end{bottomstuff}

\maketitle

\section{Introduction}
A declarative semantics is a highly desirable property for a programming
language. It offers a clean theoretical foundation for the language, allows to
prove program properties such as correctness and operational equivalence and
guarantees platform independence. Declarative programs tend to be shorter and
clearer as they contain, ideally, only information about the modeled problem and
not about control.

\emph{Constraint Handling Rules} (CHR)
\cite{DBLP:journals/lncs/Fruhwirth94,DBLP:journals/jlp/Fruhwirth98,fruehwirth09}
is a declarative committed-choice general-purpose programming language developed
in the 1990s as a portable language extension to implement user-defined
constraint solvers. Operationally, it mixes rule-based multiset rewriting over
constraints with calls to a built-in constraint solver with at least rudimentary
capabilities. It is Turing complete and it has been shown that every algorithm
can be implemented in CHR with optimal time complexity
\cite{Sneyers05thecomputational}. Hence, it makes an efficient stand-alone
general-purpose programming language.

\emph{Constraint Handling Rules with Disjunction} (CHR$^\vee$)
\cite{DBLP:conf/fqas/AbdennadherS98} extends the inherently non-deterministic
formalism of CHR with the possibility to include backtracking search and thus to
embed horn programs. It features both don't-care and don't-know
non-determinism. We can justly describe it as a multi-paradigm declarative
programming language.

Owing to its heritage in logic programming and constraint logic programming,
CHR features a declarative semantics in classical logic. We have
shown that for certain classes of programs, the classical declarative semantics
of CHR reflects the functionality of a program but poorly
\cite{DBLP:conf/cp/BetzF05}. Operationally, CHR is a state transition system
whereas the classical declarative semantics considers all states in a derivation
as logically equivalent. Hence, the directionality of the rules, the inherent
non-determinism of their execution and any change of state eludes this
declarative semantics.

\emph{Linear logic} is a sub-structural logical formalism
\cite{DBLP:journals/tcs/Girard87} that has been shown  to bear a close
relationship to concurrent committed-choice systems
\cite{DBLP:conf/elp/Miller92,DBLP:journals/iandc/FagesRS01}. It shows that it is
well-suited to model the committed-choice rules of CHR. It furthermore allows a
faithful embedding of classical logic, so we can straightforwardly embed the
constraint theory underlying the built-in constraint solver into linear logic.
Linear logic thus enables us to model the two reasoning mechanisms of CHR in a
single formalism. Moreover, it shows that we can encode CHR$^\vee$ into linear
logic in a way that preserves its characteristic dichotomy of don't-know and
don't-care non-determinism.

In this article, we propose a linear-logic semantics for CHR and CHR$^vee$ that
incorporates all the features mentioned above. We found the semantics on the
intuitionistic segment of linear logic as it suffices for our purpose while
being easier to handle than the full segment. We propose two variants of the
semantics. The first variant is based on introducing proper axioms in the
sequent calculus of linear logic. The second variant is similar to the
semantics previously published in \citeN{DBLP:conf/cp/BetzF05} and
\citeN{Betz07}. The first formulation allows for considerably
more elegant proofs, in particular of its soundness and completeness. The
second formulation allows to perform a broader range of reasoning tasks. As we
formalize and prove the equivalence of both representations, so we can use
either representation according to the respective application.

This article is structured as follows:
2
In Sect.~\ref{sec:chr}, we recall the syntax and operational semantics of CHR.
3
In Sect.~\ref{sec:ill}, we introduce the intuitionistic segment of linear logic.
4
In Sect.~\ref{sec:ll_semantics}, we develop a linear-logic semantics
for constraint handling rules, and we show its soundness and completeness with
respect to the operational semantics.
5
In Sect. \ref{sec:chrv}, we extend our semantics to CHR$^\vee$ and prove its
soundness and completeness. We show that the linear-logic semantics allows in
general for less precise reasoning over CHR$^\vee$ than over CHR. We
then introduce a well-behavedness property for CHR$^\vee$ programs that
amends this limitation.  In Sect. \ref{sec:application}, we show how our
semantics can be applied to reason about program observables as well as to
compare programs even across the boundaries of programming paradigms. In Sect.
\ref{sec:related}, we discuss related work before we conclude in Sect.
\ref{sec:conclusion}.

\section{Constraint Handling Rules}
\label{sec:chr}

In this section, we recall the syntax and the operational semantics $\oesq$ of
Constraint Handling Rules.

\subsection{The Syntax of CHR}
\label{sec:syn}

We distinguish two disjoint classes of atomic constraints: \emph{atomic built-in
constraints} and \emph{atomic user-defined constraints}. We denote the former as
$\cb$ and the latter as $\cu$, where $c_u,c_b$ are $n$-ary constraint symbols and
$\bt$ is a sequence of $n$ terms. Built-in constraints and user-defined
constraint are  possibly empty  conjunctions of their respective atomic
constraints. A conjunction of atomic constraints in general, irrespective of
their class, is called a \emph{goal}\footnote{Note that the term \emph{goal} is used in
CHR for historical reasons and does not imply that program execution is
understood as proof search.}. Empty goals and empty constraints are denoted as
$\top$.

The syntax of constraints is summarized in
Def.~\ref{def:constraint-syntax}.

\begin{definition}[Constraint Syntax]
\label{def:constraint-syntax}
Let $\cb,\cu$ denote an $n$-ary atomic built-in or user-defined constraint, respectively,
where $\bt$ is an $n$-ary sequence of terms:
\begin{tabular}{l @{\quad} r @{$\,::=\,$} l}
\\
Built-in constraint: & $\B$ & $\top \mid \cb \mid \B \wedge \B'$ \\
User-defined constraint: & $\U$ & $\top \mid \cu \mid \U \wedge \U'$ \\
Goal: & $\G$  & $\top \mid \cu \mid \cb \mid \G \wedge \G'$ \\
\\
\end{tabular}

$\top$ stands for the \emph{empty constraint} or the \emph{empty goal},
respectively. The set of built-in constraints furthermore contains at least
\emph{falsity} $\bot$, and the binary constraint $\doteq$, standing for
\emph{syntactic equality}.
For any two goals $\G,\G'$, the goal equivalence relation $\G\equiv_{G}\G'$
denotes equivalence with respect to the \emph{associativity} and
\emph{commutativity} of $\wedge$ and the neutrality of the \emph{identity
element} $\top$.
\end{definition}

Both built-in and user-defined constraints are special cases of goals. The goal
equivalence relation $\G\equiv_{G}\G'$ does not account for idempotence, thus
implicitly imposing a multiset semantics on goals. For example,
$\cu\wedge\cu\not\equiv_{G}\cu$. We denote the set of variables occurring in a
goal $\G$ as $vars(\G)$.

A CHR program is a set of rules adhering to the following definition:

\begin{definition}[Rule Syntax] \label{def:chr_rule}
\begin{longenum}
\item A CHR rule is of
the form
\[
r\ @\ H_1\setminus H_2 \Leftrightarrow G \mid B_u\wedge B_b
\]
The \emph{rule head} $H_1\setminus H_2$ consists of the \emph{kept head} $H_1$
and the \emph{removed head} $H_2$. Both $H_1,H_2$ are user-defined constraints.
At least one of them must be non-empty. The guard G is a built-in constraint.
The rule body is of the form $B_b\wedge B_u$, where $B_b$ is a built-in
constraint and $B_u$ is a user-defined constraint. $r$ serves as an identifier
for the rule.

\item The identifier $r$ is operationally irrelevant and can be omitted along
with the $@$. An empty guard $G=\top$ can be omitted along with the $\mid$. A
rule with an empty kept head $H_1$ can be written as $r\ @\ H_2 \Leftrightarrow
G \mid B_u\wedge B_b$. Such a rule is called a \emph{simplification rule}. A
rule where the removed head $H_2$ is empty can be written as $r\ @\ H_1
\Rightarrow G \mid B_u\wedge B_b$. Such a rule is called a \emph{propagation
rule}. A rule where neither $H_1$ nor $H_2$ are empty is called a
\emph{simpagation rule}.

\item A \emph{variant} of a rule~$r\ @\ H_1 \setminus H_2 \Leftrightarrow G \mid
B_u\wedge B_b$ with variables~$\bar x$ is of the form $(r\ @\ H_1 \setminus H_2
\Leftrightarrow G \mid B_u\wedge B_b)[\bar x / \bar y]$ where $\by$ is an
arbitrary sequence of pairwise distinct variables.

\item A \emph{CHR program} is a set of CHR rules.
\end{longenum}
\end{definition}

In anticipation of Section~\ref{sec:op-sem}, we point out that propagation rules
may cause trivial non-termination of programs as they do not in general
eliminate the pre-condition of their firing. Hence, precautions have to be
taken. We refer the reader to \citeN{DBLP:conf/cp/Abdennadher97} and
\citeN{DBLP:conf/iclp/DuckSBH04} for the most common approach based on keeping a
history of applied rules and to \citeN{DBLP:journals/tplp/BetzRF10} for a
more recent approach based on finite representations of infinite program states and computations.

\subsection{The Equivalence-Based Semantics $\oesq$}
\label{sec:op-sem}

In this section, we recall the operational semantics of CHR. Several
formalizations of the operational semantics  exist in the literature. We
choose the so-called \emph{equivalence-based semantics} $\oesq$ as it contains
all the elements that we represent in our linear-logic semantics while allowing
for elegant proofs of theoretical properties.

Operationally, built-in and user-defined constraints are handled separately. For
the handling of built-in constraints, CHR requires a so-called
\emph{predefined constraint handler} whereas user-defined constraints are
handled by the actual user program. We assume that the predefined solver
implements a \emph{complete} and \emph{decidable} first-order \emph{constraint theory}
$CT$ over the built-in constraints.

\begin{definition}[Constraint Theory] A constraint theory $CT$ is a decidable
theory of intuitionistic logic over the built-in constraints. We assume that it
is given as a set of formulas  of the form \[
\alpha::=\forall(\exists\bx.\B\rightarrow\exists\bx'.\B') \] called
\emph{$CT$-axioms} where $\B,\B'$ are possibly empty built-in constraints and
$\bx,\bx'$ are possibly empty sets of variables.
\end{definition}

It should be noted that defining constraint theories explicitly over
\emph{intuitionistic} rather than full classical logic is non-standard. It is,
however, an unproblematic decision because in the operational semantics only
judgements over conjunctions of positive literals are considered. Furthermore,
this decision allows us to restrict ourselves to the intuitionistic fragment of
linear logic when translating constraint theories into linear logic.

CHR itself is a transition system over equivalence classes of program states,
which are defined as follows:

\begin{definition}[CHR State]
\label{def:binary-state}
\begin{enumerate}
\item  A CHR state is a tuple of the form $S = \state{\G;\V}$ where $\G$ is a
goal called \emph{constraint store} and $\V$ is a set of variables called
\emph{global variables}.
\item For a CHR state $S=\state{\U\wedge\B;\V}$, where $\U$ is a
user-defined constraint and $\B$ is a built-in constraint, we
call
\begin{enumerate}
\item $\bl_S ::= ( vars(\U)\cup vars(\B) ) \setminus \V$\quad the
\emph{local variables} of S and
\item $\bs_S ::= \bl_S \setminus vars(\U)$\quad the \emph{strictly local
variables} of S.
\end{enumerate}
\item A \emph{variant} of a state $S = \state{\G;\V}$ with local variables
$\bl$ is a state $S'$ of the form $S'=\state{\G[\bl/\bx];\V}$, where $\bx$ is a sequence
of pairwise distinct variables that do not occur in $\V$.
\end{enumerate}
\end{definition}

The state transition system that formalizes the operational semantics
builds on the following definition of equivalence
between CHR states:

\begin{definition}[Equivalence of CHR States]
\label{def:s_equiv}

In the following, let $\U,\U'$ denote arbitrary user-defined
constraints, $\B,\B'$ built-in constraints, $\G,\G'$
goals, $\V,\V'$ sets of variables $v$ a variable and $t$ a term.
State equivalence, written as $\cdot\equiv_e\cdot$, is the smallest
equivalence relation over CHR states that satisfies all of the following
conditions:

\begin{enumerate}
\item \label{cond:se_comm} \emph{(Goal Transformation)}
\[
\G\equiv_G\G'
\quad\Rightarrow\quad
\state{\G;\V} \equiv_e \state{\G';\V}
\]
\item \label{cond:se_subst} \emph{(Equality as Substitution)}
\[
\state{\U\wedge\xet\wedge\B;\V} \equiv_e \state{\U\subxt\wedge\xet\wedge\B;\V}
\]
\item \label{cond:se_appct} \emph{(Application of CT)} Let $\bs,\bs'$ be the strictly local variables of $\state{\U\wedge\B;\V}, \state{\U\wedge\B';\V}$. If
$CT\models\exists {\bar s}.\B \leftrightarrow\exists{\bar s}'.\B'$
then:
\[
\state{\U\wedge\B;\V} \equiv_e \state{\U\wedge\B';\V}
\]
\item \label{cond:se_global} \emph{(Neutrality of Redundant Global Variables)}
\[
x\not\in vars(\G)
\quad\Rightarrow\quad
\state{\G;\{x\}\cup\V} \equiv_e \state{\G;\V}
\]
\item \label{cond:se_fail} \emph{(Equivalence of Failed States)}
For all goals $\G,\G'$ and all sets of variables $\V,\V'$:
\[
\state{\G\wedge\bot;\V} \equiv_e \state{\G'\wedge\bot;\V'}
\]
\end{enumerate}
Where there is no ambiguity, we usually write $\cdot\equiv\cdot$ rather than
$\cdot\equiv_e\cdot$.
\end{definition}

While we generally impose a multiset semantics over goals,
Definition~\ref{def:s_equiv}.\ref{cond:se_appct} implicitly restores the set
semantics for built-in constraints within states. When discussing \emph{pure
CHR} -- as opposed to its generalization CHR$^\vee$ (cf. Sect.\ref{sec:chrv}) --
we will usually consider states in the following \emph{normal form}:

\begin{definition}[Normal Form of CHR States]
\label{def:ternary-state}
A CHR state $S$ is considered in \emph{normal form} if it is of the form
$S=\state{\U\wedge\B;\V}$ where $\U$ is a user-defined
constraint called the \emph{user-defined store} and $\B$ is a
built-in constraint called the \emph{built-in store}. Such a
state is usually written in ternary notation: $\state{\U;\B;\V}$.
\end{definition}

Any state with an inconsistent built-in store is called a \emph{failed state} as
formalized in the following definiton:

\begin{definition}[Failed State]
Any CHR state $S\equiv\state{\U;\bot;\V}$ for some $\G,\V$ is called a
\emph{failed state}. We use $S_\bot=\state{\top;\bot;\emptyset}$ as the default
representative for the set of failed states.
\end{definition}

The following lemma states several properties following from
Def.~\ref{def:s_equiv} that have been presented and proven in
\citeN{Raiser2009a}:

\begin{lemma}[Properties of State Equivalence]
\label{lem:se_derived}
The following properties hold in general:
\begin{longenum}
\item \label{prop:se_rename} \emph{(Renaming of Local Variables)}
\[
\state{\U;\B;\V} \chreq \state{\U\subxy;\B\subxy;\V}
\]
for $x\not\in\V$ and $y\not\in\V$ and $y$ does not occur in $\U$ or $\B$.
\item \label{prop:se_partial} \emph{(Partial Substitution)} Let $\U\psubxt$ be
a
user-defined constraint where \emph{some} occurrences of $x$ are substituted with $t$:
\[
\state{\U;\xet,\B;\V} \equiv \state{\U\psubxt;\xet,\B;\V}
\]
\item \label{prop:se_lequiv} \emph{(Logical Equivalence)} If
\[
\state{\U;\B;\V} \equiv \state{\U';\B';\V'}
\]
then $CT\models (\exists \bl.\U\wedge \B) \leftrightarrow
(\exists\bl'.\U'\wedge \B')$,  where $\bl,\bl'$ are the local variables of
$\state{\U;\B;\V},  \state{\U';\B';\V'}$, respectively.
\end{longenum}
\end{lemma}

Lemma~\ref{lem:se_derived}.\ref{prop:se_rename} allows us to assume without loss of generality that the local variables of any two specific states are renamed apart. Concerning Lemma~\ref{lem:se_derived}.\ref{prop:se_lequiv}, note that logical
equivalence of $\exists \bl.\U\wedge \B$ and $\exists\bl'.\U'\wedge \B'$ is
a \emph{necessary} but not a \emph{sufficient} condition for state equivalence. The linear logic semantics will enable us to formulate a similar condition that is both necessary \emph{and} sufficient (cf. Sect.~\ref{sec:pt-sem}).

The task of deciding equivalence -- and more so: non-equivalence -- is not
always trivial using the axiomatic definition. We quote
Theorem~\ref{thm:sq_crit} which gives a necessary, sufficient, and decidable
criterion. It uses the following notion of \emph{matching}:

\begin{definition}[Matching of Constraints]
For user-defined constraints $\U=c_1(\bt_1)\wedge \ldots \wedge c_n(\bt_n),
\U'=c'_1(\bt'_1)\wedge \ldots \wedge c'_m(\bt'_m)$, the matching
relation $\U\doteq\U'$ holds \emph{if and only if} $n=m$ \emph{and} there exists a permutation $\sigma$
such that
\[
\bigwedge_{i=1}^n c_i(\bt_i) \doteq c'_{\sigma(i)}(\bt'_{\sigma(i)})
\]
\end{definition}

The following theorem has been published and proven in
\cite{Raiser2009a}.

\begin{theorem}[Criterion for $\equiv_e$]\label{thm:sq_crit}
\label{thm:criterion}
Consider CHR states $S = \state{\U;\B;\V}, S' = \state{\U';\B';\V}$ with local variables~$\bl,\bl'$ that have been renamed apart. Then $S\equiv S'$
\emph{if and only if}:
\[
CT\models \forall (\B \rightarrow \exists \bl'.((\U \de \U') \wedge
\B')) \wedge \forall (\B' \rightarrow \exists \bl.((\U \de \U') \land \B))
\]
\end{theorem}

We define the notion of \emph{local variables} of CHR rules, which is necessary
for the definition of the operational semantics:

\begin{definition}[Local Variables in Rules]
For a CHR rule $r\ @\ H_1\setminus H_2\Leftrightarrow G \mid B_u\wedge B_b$, we
call the set
\[
\by_r = vars(B_u, B_b, G)\setminus vars(H_1, H_2)
\]
the \emph{local variables} of $r$.
\end{definition}

The transition system constituting the operational semantics of CHR is specified in the following
definition:

\begin{definition}[Transition System of $\oesq$]
\label{def:chr-op-sem}

CHR is a state transition system over equivalence classes of CHR states defined
by the following transition rule, where $(r\ @\ H_1 \setminus H_2 \Leftrightarrow
G\mid B_u \wedge B_b)$ is a variant of a CHR rule whose local variables $\by_r$
are renamed apart from any variable in $vars(H_1,H_2,\U,\B,\V)$:
\medskip
\[
\frac{
r\ @\ H_1 \setminus H_2 \Leftrightarrow G\mid B_u \wedge B_b
\quad\quad\quad
CT\models \exists(G\wedge\B)
}{
[\state{H_1 \wedge H_2 \wedge \U;G\wedge\B;\V}]
\mapsto^r
[\state{H_1 \wedge B_c \wedge \U;G\wedge B_b\wedge\B;\V}]
}
\]

\medskip
If the applied rule is obvious from the context or irrelevant, we write
transition simply as $\mapsto$. We denote its reflexive-transitive closure as
$\mapsto^{*}$. In the following, we sometimes write $S \mapsto T$ instead of
$[S] \mapsto [T]$ to preserve clarity.
\end{definition}

The required disjointness of the local variables $\by_r$ from all variables
occurring in the pre-transition state outside $G$ enforces that fresh variables
are introduced for the local variables of the rule. When reasoning about
programs, we usually refer to the following observables:

\begin{definition}[Computables States and Constraints]
\label{def:observables}
Let $S$ be a CHR state, $\bbP$ be a program, and $CT$ be a constraint theory. We
distinguish three sets of observables:

\begin{tabular}{l @{\quad} r @{$\,::=\,$} l}
\\
Computable states:  & $\cC_{\bbP,CT}(S)$ &
$\{[T]\mid [S]\mapsto^*[T]\}$ \\
Answers: & $\cA_{\bbP,CT}(S)$ &
$\{[T]\mid [S]\mapsto^*[T]\not\mapsto\}$ \\
Data-sufficient answers: & $\cS_{\bbP,CT}(S)$ &
$\{[\state{\top;\B;\V}]\mid [S]\mapsto^*[\state{\top;\B;\V}]\}$ \\
\\
\end{tabular}

For all three sets, if the respective constraint theory $CT$ is clear from the
context or not important, it may be omitted from the  identifier of the
respective set.
\end{definition}

As the transition system does not allow transitions from an empty user-defined
store (nor from failed states), the data-sufficient answers
$\cS_{\bbP,CT}(S)$ are a subset of the answers $\cA_{\bbP,CT}(S)$ of any state
$S$. The following property follows directly:

\begin{property}[Hierarchy of Observables]
For any state $S$, program $\bbP$ and constraint theory $CT$, we have:
\[
\cS_{\bbP,CT}(S) \subseteq \cA_{\bbP,CT}(S) \subseteq \cC_{\bbP,CT}(S)
\]
\end{property}

Confluence is an important property in transition systems. We define it in the
usual manner:

\begin{definition}[Confluence]
\label{def:confluence}
A CHR program $\bbP$ is \emph{confluent} if for all states $S,T,T'$ such that
$[S]\mapsto^* [T]$ and $[S]\mapsto^* [T']$, there exists a state $T''$ such that
$[T]\mapsto^* [T'']$ and $[T']\mapsto^* [T'']$.
\end{definition}

Confluence restricts the number of possible answers to a query:

\begin{property}\label{prop:confluence-answers}
Let $\bbP$ be a confluent CHR program. Then for every CHR state $S$, we have
$|\cS_{\bbP}(S)|\in\{0,1\}$ and $|\cA_{\bbP}(S)|\in\{0,1\}$, where
$\mid\cdot\mid$ denotes cardinality.
\begin{proof}[sketch]
We assume that for some states $S,T,T'$ and some confluent program $\bbP$, we
have $S\der^* T\nder$ and $S\der^* T'\nder$ and $[T]\neq[T']$. Applying
Def.~\ref{def:confluence} leads to a contradiction.
\end{proof}
\end{property}

A necessary, sufficent and decidable criterion for confluence has been given in
\citeN{Abdennadher96onconfluence}. Example~\ref{example:leq} presents a
standard CHR example program to illustrate our definitions.

\begin{example}
\label{example:leq}

The following program implements a solver for the (user-defined) partial-order
constraint $\leq$. Rule $r_I$ implements idempotence of identical constraints,
$r_R$ implements reflexivity, $r_S$ symmmetry and $r_T$ transitivity of the
partial-order relation:

\medskip
\[
\begin{array}{lclcl}
r_I & @ & x\leq y \wedge x\leq y & \Leftrightarrow & x\leq y \\
r_R & @ & x\leq x & \Leftrightarrow & \top \\
r_S & @ & x\leq y \wedge y\leq x & \Leftrightarrow & x=y \\
r_T & @ & x\leq y \wedge y\leq z & \Rightarrow & x\leq z \\
\end{array}
\]

\medskip
The following is a sample derivation, starting from an initial state
$S_0=\state{a\leq b \wedge b\leq c \wedge c\leq a;\top;\{a,b,c\} }$. According
to the usual practice, all variables occurring in the initial state are global.
Equivalence transformations are stated explicitly:

\setcounter{equation}{0}
\begin{align}
& \state{a\leq b \wedge  b\leq c \wedge c\leq a; \top;\{a,b,c\}}
\label{eq:leq_1}\\
\equiv\: & \state{ \underline{x\leq y} \wedge \underline{y\leq z} \wedge c\leq
a; x\doteq a \wedge y\doteq b \wedge z\doteq c;\{a,b,c\}} \nonumber \\
\mapsto^{r_T}\: & \state{\underline{x\leq z} \wedge x\leq y \wedge y\leq z \wedge c\leq a;
x\doteq a \wedge y\doteq b \wedge z\doteq c;\{a,b,c\}}\nonumber \\
\equiv\: & \state{a\leq c \wedge a\leq b \wedge b\leq c \wedge c\leq a;\top;\{a,b,c\}}
\label{eq:leq_2} \\
\equiv\: & \state{\underline{x\leq y} \wedge \underline{y\leq x} \wedge a\leq b \wedge b\leq
c;x\doteq a \wedge y\doteq c;\{a,b,c\}}\nonumber \\
\mapsto^{r_S}\: & \state{a\leq b \wedge b\leq
c;\underline{x\doteq y} \wedge x\doteq a \wedge y\doteq c;\{a,b,c\}}\nonumber \\
\equiv\: & \state{a\leq b \wedge b\leq c;a\doteq c;\{a,b,c\}} \label{eq:leq_3} \\
\equiv\: & \state{\underline{x\leq y} \wedge \underline{y\leq
x};x\doteq a \wedge y\doteq b \wedge a=c;\{a,b,c\}}\nonumber \\
\mapsto^{r_S}\: & \state{\top;\underline{x\doteq y} \wedge x\doteq a \wedge y\doteq b \wedge a\doteq c;\{a,b,c\}}\nonumber \\
\equiv\: & \state{\top;a\doteq b \wedge  a\doteq c;\{a,b,c\}} \label{eq:leq_4}
\end{align}

Usually, we do not make equivalence transformations explicit and list only
states where local variables are eliminated as far as possible such as the
labeled states~(\ref{eq:leq_1})-(\ref{eq:leq_4}). The derivation is then reduced
to:
\setcounter{equation}{0}
\begin{align}
& \state{\underline{a\leq b} \wedge \underline{b\leq c} \wedge c\leq a;\top;\{a,b,c\}} \\
\mapsto^{r_T}\: & \state{\underline{a\leq c} \wedge a\leq b \wedge b\leq c \wedge \underline{c\leq
a};\top;\{a,b,c\}} \\
\mapsto^{r_S}\: & \state{\underline{a\leq b} \wedge \underline{b\leq c};a\doteq c;\{a,b,c\}}
\\
\mapsto^{r_S}\: & \state{\top;a\doteq b \wedge a\doteq c;\{a,b,c\}}
\end{align}
\end{example}

With respect to our observables, we have:
\[
\cS_{\bbP,CT}(S_0) = \cA_{\bbP,CT}(S_0) =
\{[\state{\top;a\doteq b\wedge a\doteq c;\{a,b,c\}}]\}
\]
The set $\cC_{\bbP,CT}(S_0)$ is infinite
as the operational semantics $\oesq$ allows potentially unlimited applications of $r_T$.

\section{Intuitionistic Linear Logic}
  \label{sec:ill}

Linear logic was introduced by
\citeN{DBLP:journals/tcs/Girard87}. Unlike classical logic, linear logic does
not allow free copying or discarding of assumptions. It furthermore features a fine
distinction between internal and external choice and a faithful embedding of
classical logic. In this section, we recall the intuitionistic fragment of linear
logic, which is easier to handle than the full fragment but sufficient for
our declarative semantics. It allows for a straightforward, faithful embedding of
intuitionistic logic.

\subsection{Definition}

We will give the formal definition in terms of a \emph{sequent calculus}. The
calculus is based on binary sequents of the form \[
  \Gamma \vdash \alpha
\] where $\Gamma$ is a multiset of formulas (written without braces) called
\emph{antecedent} and $\alpha$ is a formula called \emph{consequent}. A sequent
$\Gamma \vdash \alpha$ represents the fact that assuming the formulas in
$\Gamma$, we can conclude $\alpha$. A \emph{proof tree} -- or simply:
\emph{proof} -- is a finite labeled tree whose nodes are labeled with
sequents such that the relationship between every sequent node and its direct
children corresponds to one of the inference rules of the calculus. We
distinguish a special set of sequents called \emph{axioms}. A proof tree is
called \emph{complete} if all its leaves are axioms. We call a sequent
$\Gamma\vdash\alpha$ \emph{valid} if there exists a complete proof tree $\pi$
with $\Gamma\vdash\alpha$ at the root.

The following two structural rules are common to many logical systems. They establish
reflexivity and a form of transitivity of the judgement relation.

\[
\infer[(Identity)]{\alpha\vdash\alpha}{}
\quad\quad
\infer[(Cut)]{\Gamma,\Delta\vdash\beta}
    {\Gamma\vdash\alpha&\alpha,\Delta\vdash\beta}
\]

The tokens of (intuitionistic) linear logic are commonly considered as
representing \emph{resources} rather than \emph{truths}. This terminology
reflects the fact that assumptions may not be copied nor discarded freely in
linear logic, but must be used exactly once. From a different point of view, we
might say that linear logic \emph{consumes} assumptions in judgements and is
aware of their \emph{multiplicities}.

\emph{Multiplicative conjunction} is distinguished from classical or
intuitionistic conjunction as it lacks idempotence. Hence, $\alpha\x\beta$
represents \emph{exactly} one instance of $\alpha$ and one instance of $\beta$.
The formula $\alpha$ is not equivalent to $\alpha\x\alpha$. Multiplicative
conjunction is introduced by the following inference rules:

\[
\infer[(L\x)]{\Gamma,\alpha\x\beta\vdash\gamma}
    {\Gamma,\alpha,\beta\vdash\gamma}
\quad\quad
\infer[(R\x)]{\Gamma,\Delta\vdash\alpha\x\beta}
    {\Gamma\vdash\alpha&\Delta\vdash\beta}
\]

The constant $\lone$ represents the empty resource and is consequently the
neutral element with respect to multiplicative conjunction.

\[
\infer[(L\lone)]{\Gamma,\lone\vdash\alpha}{\Gamma\vdash\alpha}
\quad\quad
\infer[(R\lone)]{\vdash \lone}{}
\quad\quad
\]

\emph{Linear implication $\lp$} allows the application of \emph{modus ponens}
where the preconditions of a linear implication are consumed on application. For
example, the sequent $\alpha\x(\alpha\lp \beta)\vdash \beta$ is valid whereas
$\alpha\x(\alpha\lp \beta)\vdash \alpha\x \beta$ is not. The following inference
rules introduce $\lp$:

\[
\infer[(L\lp)]{\Gamma,\alpha\lp\beta,\Delta\vdash\gamma}
    {\Gamma\vdash\alpha&\beta,\Delta\vdash\gamma}
\quad\quad
\infer[(R\lp)]{\Gamma\vdash\alpha\lp\beta}
    {\Gamma,\alpha\vdash\beta}
\]

The \textit{! (``bang'') modality} marks stable facts or unlimited resources,
thus recovering propositions in the classical (or intuitionistic) sense. Like an
classical proposition, a \emph{banged} resource may be freely copied or
discarded. Hence, $!\alpha\x !(\alpha\lp \beta)\vdash \, !\alpha\x !\beta$ is a
valid sequent. Four inference rules introduce the bang:

\[
\infer[(R!)]{!\Gamma\vdash!\alpha}
    {!\Gamma\vdash\alpha}
\quad\quad
\infer[(Dereliction)]{\Gamma,!\alpha\vdash\beta}
    {\Gamma,\alpha\vdash\beta}
\]
\[
\infer[(Contraction)]{\Gamma,!\alpha\vdash\beta}
    {\Gamma,!\alpha,!\alpha\vdash\beta}
\quad\quad
\infer[(Weakening)]{\Gamma,!\alpha\vdash\beta}
    {\Gamma\vdash\beta}
\]

\begin{example}
We can model the fact that one cup of coffee ($c$) is one euro ($e$) as $!(e\lp
c)$. A ``bottomless cup'' is an offer including an unlimited number of refills.
We assume that any natural number of refills is possible. We model this as
$!(e\lp\bang c)$. From this, we may judge that it is possible to get two cups of
coffee for one euro: $!(e\lp\bang c)\vdash e\lp c\x c$.
Fig.~\ref{fig:coffee-example} gives an examplary proof tree, proving this
judgement.

\begin{figure}
\label{fig:coffee-example}
\begin{center}
\[
\infer[(Dereliction)]
{
  !(e\lp\bang c)\vdash e\lp c\x c
}
{
  \infer[(R\lp)]
  {
    e\lp\bang c\vdash e\lp c\x c
  }
  {
    \infer[(L\lp)]
    {
      e\lp\bang c,e\vdash c\x c
    }
    {
      \infer[(Identity)]
      {
        e\vdash e
      }
      {
      }
    &
      \infer[(Contraction)]
      {
        \bang c\vdash c\x c
      }
      {
        \infer[(R\x)]
        {
          \bang c,\bang c\vdash c\x c
        }
        {
          \infer[(Dereliction)]
          {
            \bang c\vdash c
          }
          {
            \infer[(Identity)]
            {
              c\vdash c
            }
            {
            }
          }
        &
          \infer[(Dereliction)]
          {
            \bang c\vdash c
          }
          {
            \infer[(Identity)]
            {
              c\vdash c
            }
            {
            }
          }
        }
      }
    }
  }
}
\]
\end{center}
\caption{A sample proof tree}
\end{figure}
\end{example}

In classical (and intuitionistic) logic, internal choice is an aspect of
conjunction, as exemplified by the judgement $\alpha\wedge\beta\vdash\alpha$.
This is inherited by the \textit{additive conjunction $\&$} of linear logic. The
formula $\alpha\&\beta$ expresses a choice between $\alpha$ and $\beta$, i.e. the
sequents $\alpha\&\beta\vdash\alpha$ and $\alpha\&\beta\vdash\alpha$ are valid,
but $\alpha\&\beta\vdash A\x B$ is not.

\[
\infer[(L\&_1)]{\Gamma,\alpha\&\beta\vdash\gamma}
    {\Gamma,\alpha\vdash\gamma}
\quad\quad
\infer[(L\&_2)]{\Gamma,\alpha\&\beta\vdash\gamma}
    {\Gamma,\beta\vdash\gamma}
\quad\quad
\infer[(R\&)]{\Gamma\vdash\alpha\&\beta}
    {\Gamma\vdash\alpha&\Gamma\vdash\beta}
\]

The $\ltop$ (``top'') is the resource that all other resources can be mapped to,
i.e. for every $\alpha$, the implication $\alpha\lp\ltop$ is a tautology. It is
hence the neutral element with respect to additive conjunction.

\[
\infer[(R\ltop)]{\Gamma \vdash \ltop}{}
\quad\quad
\]

External choice is an aspect of classical (and intuitionistic) disjunction. In
linear logic, it is represented by the \textit{additive disjunction $\oplus$}.
Analogous to classical logic, $\alpha\oplus\beta\vdash\alpha$ is not valid.
However, $!(\alpha\lp\gamma),!(\beta\lp\gamma), \alpha\oplus\beta\vdash\gamma$ is
valid.

\[
\infer[(L\oplus)]{\Gamma,\alpha\oplus\beta\vdash\gamma}
    {\Gamma,\alpha\vdash\gamma&\Gamma,\beta\vdash\gamma}
\quad\quad
\infer[(R\oplus_1)]{\Gamma\vdash\alpha\oplus\beta}
    {\Gamma\vdash\alpha}
\quad\quad
\infer[(R\oplus_2)]{\Gamma\vdash\alpha\oplus\beta}
    {\Gamma\vdash\beta}
\]

Analogous to falsity in the classical sense, \emph{absurdity} $\lzero$ is
a constant that yields every other resource. It is the neutral element with
respect to $\oplus$.

\[
\infer[(L\lzero)]{\lzero\vdash\alpha}{}
\]
\begin{example} We assume that, besides coffee, the cafeteria offers also pie
($p$) at the price of one euro per piece: $!(e\lp p)$. We infer that for one
euro, we have the choice between an arbitrary amount of coffee and a piece of
pie: $!(e\lp\bang c),!(e\lp p)\vdash e\lp( !c\& p)$. Let us furthermore assume that rather than
with euros, we can also pay with dollars ($d$) at a $1:1$ ratio: $!(d\lp\bang
c),!(d\lp p)$. We may infer either one of \emph{one dollar} or \emph{one euro}
buys us a choice between an arbitrary amount of coffee and one pie.:
\[!(e\lp\bang c),!(e\lp p),!(d\lp\bang c),!(d\lp p)\vdash (e\oplus d)\lp( !c\& p).\]
\end{example}

We can extend intuitionistic linear logic into a first-order system with the
quantifiers $\exists$ and $\forall$. Their introduction rules are the same as in
classical logic. In the following rules, $t$ stands for an arbitrary term whereas $a$ stands for a variable
that is not free in $\Gamma$, $\alpha$ or $\beta$:

\vspace{1mm}
\[ \infer[(L\forall)]{\Gamma,\forall x.\alpha\vdash\beta}
            {\Gamma,\alpha[x/t]\vdash\beta}
\quad\quad \infer[(R\forall)]{\Gamma\vdash\forall x.\beta}
            {\Gamma\vdash\beta[x/a]}
\]
\vspace{1mm}
\[ \infer[(L\exists)]{\Gamma,\exists x.\alpha\vdash\beta}
            {\Gamma,\alpha[x/a]\vdash\beta}
\quad\quad \infer[(R\exists)]{\Gamma\vdash\exists x.\beta}
            {\Gamma\vdash\beta[x/t]}
\]

\subsection{Properties of Intuitionistic Linear Logic}

The resulting first-order system allows for a faithful embedding of
intuitionistic first order logic. This is widely considered one of the most important
features of linear logic. The following translation from intuitionistic logic
into intuitionistic linear logic is a variant of a translation proposed by
\citeN{DBLP:journals/mscs/Negri95}:

\begin{definition}
\label{def:negri}
$(\cdot)^*$ is a translation from formulas of intuitionistic logic to
formulas of intuitionistic linear logic, recursively defined by the following
rules:
\begin{center}
\begin{tabular}{r  @{\hspace{1mm}::=\hspace{1mm}} l}
$p(\bar t)^*$             & $!p(\bar t)$ \\
$(\bot)^*$                & $\lzero$ \\
$(\top)^*$                & $\lone$ \\
$(A \wedge B)^*$         & $A^* \x B^*$ \\
$(A \vee B)^*$             & $A^* \,\oplus\, B^*$ \\
$(A \rightarrow B)^*$     & $!(A^*\multimap B^*)$ \\
$(\forall x. A)^*$        & $!\forall x.(A^*)$ \\
$(\exists x. A)^*$        & $\exists x.(A^*)$ \\
\end{tabular}
\end{center}
\end{definition}

$p(\bar t)$ stands for an atomic proposition. The definition is extended to
sets and multisets of formulas in the obvious manner. It has been proven in
\citeN{DBLP:journals/mscs/Negri95} that an intuitionistic sequent
$\left(\Gamma\vdash_{IL}\alpha\right)$ is valid \emph{if and only if}
$\left(\Gamma^*\vdash_{ILL}\alpha^*\right)$ is valid in linear logic.

We distinguish two sorts of axioms in the sequent calculus. The $(Identity)$
axiom and the constant axioms $(L\lone)$, $(R\lone)$, $(L\lzero)$ and $(R\ltop)$
constitute the \emph{logical axioms} of intuitionistic linear logic. All axioms
we add to the system on top of these are called \emph{non-logical axioms} or
\emph{proper axioms}. We usually use the letter $\Sigma$ to denote the set of
proper axioms.

We express the fact that a judgement $\Gamma\vdash\alpha$ is provable using a
non-empty set $\Sigma$ of proper axioms by indexing the judgement relation with
the set of proper axioms: $\vdash_\Sigma$.

\begin{definition}[Linear-Logic Equivalence]
	\label{def:ll-equiv}
\begin{longenum}
 \item We call two linear-logic formulas $\alpha,\beta$ \emph{logically
 equivalent} if both $\alpha\vdash\beta$ and $\beta\vdash\alpha$ are provable. We write
this as $\alpha\equivll\beta$.
 \item For any set of proper axioms $\Sigma$, we call two linear-logic formulas
 $\alpha,\beta$ \emph{logically equivalent modulo $\Sigma$} if both
 $\alpha\vdash_\Sigma\beta$ and $\beta\vdash_\Sigma\alpha$ are provable. We
 write this as $\alpha\dashv\vdash_\Sigma\beta$.
\end{longenum}
\end{definition}

As a well-behaved logical system, linear logic features a cut-elimination
theorem \cite{DBLP:journals/tcs/Girard87}:

\begin{theorem}[Cut Elimination Theorem]
\label{thm:cut-elim}
\begin{longenum}
  \item \label{prop:cut-free} If a sequent $\Gamma\vdash\alpha$ has a proof $\pi$
  that does not contain any proper axioms, then it has a proof $\pi'$ that
  contains neither proper axioms nor the $(Cut)$ rule.
  \item \label{prop:cut-red} If a sequent $\Gamma\vdash_\Sigma\alpha$ has a proof
  $\pi$ containing proper axioms, then it has a proof $\pi'$ where the $(Cut)$
  rule is only used at the leaves such that one of its premises is an axiom.
\end{longenum}
\end{theorem}

A proof without any applications of $(Cut)$ is called \emph{cut-free}. A proof where
$(Cut)$ is only applied at the leaves is called \emph{cut-reduced}.

A important consequence of cut elminiation is the \emph{subformula property}. We
quote a weak formulation of the property, which will suffice for our purpose:
Every formula $\alpha$ in a cut-free proof of a sequent $\Gamma\vdash\beta$ is a
subformula of either $\Gamma$ or $\beta$, modulo variable renaming. In a
cut-reduced proof of a sequent $\Gamma\vdash_\Sigma\beta$, every formula $\alpha$
 is a subformula of $\Gamma$ or $\beta$, modulo variable renaming, or there
exists a proper axiom $(\Delta\vdash\gamma)\in\Sigma$ such that $\alpha$ is a
subformula of $\Delta$ or $\gamma$, modulo variable renaming.

\section{A Linear-Logic Semantics for CHR}
\label{sec:ll_semantics}

In this section, we motivate and develop the linear-logic semantics for
Constraint Handling Rules. We firstly recall the classical declarative semantics
in Sect.~\ref{sec:classical-analysis}. Then we motivate and present a
linear-logic semantics based on proper axioms in Sect.~\ref{sec:pt-sem}. We will
henceforth call this the \emph{axiomatic} linear-logic semantics for CHR. Its
soundness with respect to the operational semantics is shown in
Sect.~\ref{sec:pt-soundness}. We continue in Sect.~\ref{sec:entailment} by
introducing the notion of \emph{state entailment}, which we use to formulate and
prove the completeness of our semantics in Sect.~\ref{sec:pt-completeness}.
Finally, in Sect.~\ref{sec:enc-sem}, we show an alternative linear-logic
semantics that encodes programs and contraints theories into linear logic.

\subsection{Analysis of the Classical Declarative Semantics}
  \label{sec:classical-analysis}

CHR is founded on a classical declarative semantics, which is reflected in its
very syntax. In this section, we recall the classical declarative semantics and
discuss its assets and limitations.

In the following, $\exists_{-\bar x}$ stands for existential quantification of
all variables except those in $\bar x$, where $\bar x$ is a set of variables.
The classical declarative semantics is given in the following table, where
$(\cdot)^\dagger$ stand for translation to classical logic:

\begin{center}
\begin{tabular}{l @{\hspace{1em}} l @{\hspace{1mm} ::= \hspace{1mm}} l}
States:&
$\state{\U;\B;\V}^\dagger$ &
$\exists_{-\V}.(\U\wedge \B)$ \\
Rules: &
$(r\ @\ H_1\setminus H_2 \Leftrightarrow G \mid B)^\dagger$ &
$\forall \left(G\rightarrow (H_1 \rightarrow (H_2 \leftrightarrow \exists\bar
y_r.B)) \right)$
\\
Programs: &
$\{R_1, ... ,R_m\}^\dagger$ &
$R_1^\dagger \wedge \ldots \wedge R_m^\dagger$\\
\end{tabular}
\end{center}

$\by_r$ denotes the local variables of the respective rule. The following lemma
-- cited from \citeN{Fruhwirth03} -- establishes the relationship between the
logical readings of programs, constraint theories and states:

\begin{lemma}[(Logical Equivalence of States)]
Let $\bbP$ be a CHR program and $S$ be a state. Then for all computable states
$T_1$ and $T_2$ of $S$, the following holds: $\bbP^\dagger, CT\models
\forall(T_1^\dagger \leftrightarrow T_2^\dagger)$.
\end{lemma}

The declarative semantics of CHR must be distinguished from LP languages and
related paradigms as CHR is not based on the notion of \emph{execution as proof
search}. Declaratively, execution of a CHR program means stepwise transformation
of the information contained in the state under logical equivalence as defined
by the program's logical reading $\bbP^\dagger$ and the constraint theory $CT$.
Founding CHR on such a declarative semantics is an obvious choice for several
reasons:

Firstly, the notion of \emph{execution as proof search} naturally implies a
notion of \emph{search}. This stands in contrast to the committed-choice execution of CHR.
Furthermore, the forward-reasoning approach faithfully captures the one-sided
variable matching between rule heads and constraints in CHR, as opposed to
unification. For example, a CHR state $\state{p(x);\top;\emptyset}$ (where $x$
is a variable) does not match with the rule head $(p(a)\Leftrightarrow\ldots)$ (where
$a$ is a constant) just as we cannot apply modus ponens on a fact $\exists
x.p(x)$ and an implication $(p(a)\rightarrow\ldots)$. In contrast, an LP goal
$p(x)$ would be unified with a rule head $(p(a)\leftarrow\ldots)$, accounting
for the fact that application of the rule might lead to a proof of an instance
of $p(x)$.

There are, however, several limitations to the classical declarative semantics
of CHR, which shall be discussed in the following:

\paragraph*{Directionality} One limitation lies in the fact that the classical
declarative semantics does not capture the inherent directionality of CHR rules.
Rather, all states within a computation are considered logically equivalent.
Consider e.g. the minimal CHR program
\[
  a \Leftrightarrow b
\]
In this program, we can compute a state $\state{b;\top;\emptyset}$ from a state
$\state{a;\top;\emptyset}$ but not vice versa. This is not captured in its
logical reading $(a\leftrightarrow b)$ which e.g. implies $b\rightarrow a$.
The classical declarative semantics cannot be used e.g. to show that the state
$\state{a;\top;\emptyset}$ is not a computable state
$\state{b;\top;\emptyset}$.

\paragraph*{Dynamic Change} Any program state that does not only
contain declarative information about a supposedly static object world but also
meta-information about the program state eludes the semantics. Consider the
following program which computes the minimum of a set:
\[
  min(x),min(y)\Leftrightarrow x\leq y \mid min(x)
\]
On a fixed-point execution, the program correctly computes the minimum of all
arguments of $min$ constraints found in the store at the beginning of the
computation. Its logical reading, however, is unhelpful at best:
\[
	\forall x,y.
	x\leq y \rightarrow(min(x)\wedge min(y)\leftrightarrow min(x))
\]

\paragraph*{Deliberate Non-Determinism} Any program that makes deliberate use of
the inherent non-determinism of CHR has a misleading declarative semantics as
well. Consider the following program, which simulates a coin throw in an appropriate
probabilistic semantics of CHR (cf. \citeN{DBLP:journals/entcs/FruhwirthPW02}).
(Note that $coin$ is a variable, $head$ and $tail$ are constants.)
\[
  \begin{array}{lcl}
    throw(coin) & \Leftrightarrow & coin \de head \\
    throw(coin) & \Leftrightarrow & coin \de tail
  \end{array}
\]
The logical reading of this program implies $\forall coin.(coin\de
head\leftrightarrow coin\de tail)$. From this follows $head\de tail$ and -- since
$head$ and $tail$ are distinct constants -- falsity $\bot$. The program's logical
reading is thus inconsistent, trivially implying anything.

\paragraph*{Multiplicities} Finally, while CHR faithfully keeps track of the
multiplicities of constraints, this aspect eludes the classical semantics.
Consider the idempotence rule from Example~\ref{example:leq}, which removes
multiple occurrences of the same constraint:
\[
    r_I \ @ \ x\leq y,x\leq y \ \Leftrightarrow \ x\leq y
\]
The logical reading of this rule is a tautology, falsely suggesting that
the rule is redundant:
\[
	\forall x,y.(x\leq y
	\wedge x\leq y \ \leftrightarrow \ x\leq y)
\]
In conclusion, the classical declarative semantics is a powerful tool to prove
the soundness and a certain notion of completeness of any program whose states
contain only model information about a static object world and no explicit
meta-information. It faithfully captures the logical theory behind those
programs. However, it is not adequate to capture the logic behind programs that
deal with any form of meta-information, make deliberate use of non-determinism or
rely on the multiplicities of constraints. As it does not capture the inherent
directionality of CHR rules, it is not suitable to prove safety conditions, i.e.
to show that a certain intermediate or final state cannot be derived from a
certain initial state.

\subsection{The Axiomatic Linear-Logic Semantics for CHR}
  \label{sec:pt-sem}

Our linear-logic semantics is based on two observations: Firstly, the difference in behaviour
between built-in and user-defined constraints in CHR resembles the difference
between linear and banged atoms in linear logic. Secondly, the application of
simplification rules on user-defined constraints resembles the application
of modus ponens in linear logic.

Building on the first observation, we define an adequate representation of CHR
constraints in linear logic. Translation to linear logic will be denoted as
$(\cdot)^L$. For atomic constraints, the choice is obvious:
\begin{align*}
	\cu^L & ::= \cu \\
	\cb^L & ::= \bang\cb
\end{align*}
Classical conjunction is mapped to multiplicative conjunction for both built-in
and user-defined constraints.
\[
  (\G \wedge \G')^L ::= \G^L\otimes {\G'}^L
\]
This mapping is motivated by the fact that
multiplicative conjunction is aware of multiplicities and has no notion of
weakening, thus capturing the multiset semantics of user-defined constraints. For
any built-in constraint $\B$, the mapping equals the translation quoted in
Def.~\ref{def:negri}: $\B^L=\B^*$. Accordingly, we map the empty goal $\top$ to
$\lone$ and falsity $\bot$ to $\lzero$. The translation of CHR states is
analogous to the classical case:
\[
  \state{\U;\B;\V}^L ::= \exists_{-\V}.\U^L \otimes \B^L
\]
The translation of
constraints, goals and states is summed up in Fig.~\ref{fig:states}.

\begin{figure}
	\label{fig:states}
	\begin{center}
	\fbox{
	\begin{tabular}{l @{\quad} r @{$\,::=\,$} l}
	\textrm{Atomic built-in constraints:} & $\cb^L$ & $!\cb$ \\
	\textrm{Atomic user-defined constraints:} & $\cu^L$ & $\cu$ \\
	\textrm{Falsity:} & $\bot^L$ & $\lzero$ \\
	\textrm{Empty constraint/goal:} & $\top^L$ & $\lone$ \\
	\textrm{Constraints/goals:} &
		$(\G_1\wedge\G_2)^L$ & $\G^L_1\otimes\G^L_2$\\
	\textrm{States:} & $\state{\U;\B;\V}^L$ & $\exists_{-\V}.\U^L\x\B^L$ \\
	\end{tabular}
	}
	\caption{Translation of constraints, goals and states}
	\end{center}
\end{figure}

\paragraph*{Proper axioms} The constraint theory $CT$, the interaction between
equality constraints (which are by definition built-in) and user-defined
constraints, and programs are translated to proper axioms. Firstly, we define a
set of proper axioms encoding the constraint theory as well as modelling the
interaction between equality $\doteq$ and user-defined constraints.

\begin{definition}[($\Sct$)] For built-in constraints $\B,\B'$
and sets of variables $\bx, \bx'$ such that $CT\models \exists\bx.\B\rightarrow
\exists\bx'.\B'$, the following is a proper axiom:
\[
\exists\bx.\B^L\vdash \exists\bx'.\B'^L
\]
We denote the set of all such axioms as $\Sct$.
\end{definition}

\begin{definition}[($\Seq$)]
If $\cu$ is an n-ary user-defined constraint and $t_j,u$ are terms such that
$t_j$ is the $j$th argument of $\cu$ then
\[
	c_u(...,t_j,...) \x !(t_j\doteq u)
	\vdash c_u(...,u,...) \x !(t_j\doteq u)
\]
is a proper axiom. We denote the set of all such axioms as $\Seq$.
\end{definition}

\begin{definition}[($\Sp$)] If $r\ @\ H_1\setminus H_2\Leftrightarrow G\mid B_b
\wedge B_u$ is a variant of a rule with local variables $\by_r$, the sequent
\[
  H_1^L\x H_2^L\x G^L \vdash H_1^L \x \exists\by_r.(B_b^L\x B_u^L\x G^L)
\]
is a proper axiom. For a program $\bbP$, we denote the set of all axioms derived
from its rules as $\Sp$.
\end{definition}

The existential quantification of the local variables $\by_r$ corresponds to the
fact that these variables are by definition disjoint from
$vars(H_1,H_2,\U,\B,\V)$, assuring that fresh variables are introduced for the
local variables of the rule. Fig.~\ref{fig:chr-axiomatic-semantics} sums up the
three sets of proper axioms, represented as inference rules.

\begin{figure}
\label{fig:chr-axiomatic-semantics}
\begin{center}
\fbox{
	\begin{tabular}{c}
	$
		\infer[(\Sct)]
		{
			\exists\bx.\B^L\vdash \exists\bx'.\B'^L
		}
		{
			CT\models \exists\bx.\B\rightarrow\exists\bx'.\B'
		}
		\quad\quad\quad
		\infer[(\Seq)]
		{
			c_u(...,t_j,...) \x !(t_j\doteq u)
			\vdash c_u(...,u,...) \x !(t_j\doteq u)
		}
		{
			\phantom{\top}
		}
	$
	\\
	\\
	$
		\infer[(\Sp)]
		{
			H_1^L\x H_2^L\x G^L \vdash H_1^L \x \exists\by_r.(B_b^L\x B_u^L\x G^L)
		}
		{
			(r\ @\ H_1\setminus H_2\Leftrightarrow G\mid B_b\wedge B_u)\subxy\in\bbP
		}
	$
	\end{tabular}
}
\end{center}
\caption{The axiomatic linear-logic semantics, represented as inference rules}
\end{figure}

In anticipation of the soundness theorem presented in
Sect.~\ref{sec:pt-soundness}, we give an example of a CHR derivation and show
that it corresponds to a valid linear logic judgement:

\begin{example}
	\label{example:pt-sem-soundness}
Let $\bbP$ be the partial-order constraint solver from
Example~\ref{example:leq} and let $CT$ be a minimal constraint theory. We
observe that under $\bbP$, we have:
\[
	[\state{3\leq a;a=3;\emptyset}]\mapsto^*[\state{\top;\top;\emptyset}]
\]
This corresponds to the judgement
$\state{3\leq a;a=3;\emptyset}^L\vdash_\Sigma\state{\top;\top;\emptyset}^L$ or
$\exists a.(3\leq a\x !a=3)\vdash_\Sigma \lone$, respectively, where
$\Sigma=\Sct\cup\Seq\cup\Sp$. The following is a proof of this judgement:
\[\scriptsize
	\infer[(L\exists)]
	{
	\exists a.(3\leq a\x !a\doteq3)\vdash \lone
	}
	{
		\infer[(Cut)]
		{
		3\leq x\x !x\doteq3\vdash \lone
		}
		{
			\infer[(Cut)]
			{
			3\leq x\x !x\doteq3\vdash \top\x !x\doteq3
			}
			{
				\infer[(\Sct)]
				{
				3\leq x\x !x\doteq3\vdash x\leq x\x !x\doteq3
				}
				{
				}
			&
				\infer[(L\otimes)]
				{
				x\leq x\x !x\doteq3\vdash \lone\x !x\doteq3
				}
				{
					\infer[(R\otimes)]
					{
					x\leq x, !x\doteq3\vdash \lone\x !x\doteq3
					}
					{
						\infer[(\Sp)]
						{
						x\leq x\vdash\lone
						}
						{
						}
					&
						\infer[(Identity)]
						{
						!x\doteq3\vdash !x\doteq3
						}
						{
						}
					}
				}
			}
		&
			\infer[(\Sct)]
			{
			\lone\x !x\doteq3\vdash\lone
			}
			{
			}
		}
	}
\]
The sequent $\lone\x !x\doteq3\vdash\lone$ is a tautology and as
such could be derived without proper axioms, but it is also trivially included in
$\Sct$.
\end{example}

While the soundness result for our semantics is straightforward, defining
completeness is not quite as simple. Consider the following example:

\begin{example}
	\label{example:pt-sem-completeness-problem}
In the proof tree given in Example~\ref{example:pt-sem-soundness} we use the
following proper axiom from $\Sct$:
\[
	\lone\x !x\doteq3\vdash\lone
\]
This implies:
\[
	\state{\top;x\doteq3;\{x\}}^L\vdash_\Sigma\state{\top;\top;\emptyset}^L
\]
We observe, however, that
$\state{\top;x\doteq3;\{x\}}\mapsto^{*}\state{\top;\top;\emptyset}$ is untrue.
\end{example}

In the following section, we prove the soundness of our semantics. In
Sect.~\ref{sec:entailment}, we develop the notion of \emph{state entailment}. We
will apply this notion to specify and prove a completeness result in
Sect.~\ref{sec:pt-completeness}.

\subsection{Soundness of the Linear Logic Semantics}
\label{sec:pt-soundness}

In this section, we prove the soundness of the axiomatic linear-logic
semantics for CHR with respect to the operational semantics.

\begin{lemma}[($\equiv_e\Rightarrow\vdash$)]
  \label{lemma:sq-ll}
  Let $CT$ be a constraint theory and $\Sigma=\Sct\cup\Seq$. For arbitrary CHR
  states $S,T$, we have: \[
 S\equiv_e T \Rightarrow S^L\vdash T^L
  \]
\begin{proof}[sketch] We prove that state equivalence $S\equiv_e T$
implies linear judgement $S\vdash T$ by showing that every of the conditions
given for $S\equiv_e T$ in Def.~\ref{def:s_equiv}
implies $S\vdash T$: Def.~\ref{def:s_equiv}.\ref{cond:se_comm}
implies linear judgement since multiplicative conjunction is associatice, commutative and invariant w.r.t.
  $\top$, thus corresponding to goal equivalence. For
  Def.~\ref{def:s_equiv}.\ref{cond:se_subst}, linear judgement is guaranteed, as
  $\Seq$ allow us to prove $\exists_{-\V}.\U\x\xet\x\B
  \vdash_\Sigma \exists_{-\V}.\U\subxt\x\xet\x\B$. For
  Def.~\ref{def:s_equiv}.\ref{cond:se_appct}, it is similarly guaranteed by $\Sct$.
  Def.~\ref{def:s_equiv}.\ref{cond:se_global} implies linear judgement since
  the addition or removal of a global variable not occurring
   in a state does not change the logical reading of the state.
  W.r.t. Def.~\ref{def:s_equiv}.\ref{cond:se_fail}, linear judgement holds since
  $\varphi\x 0\vdash\psi$ is valid for any $\varphi,\psi$.
All the above arguments can be shown to apply in the reverse direction as well,
thus proving compliance with the implicit symmetry of $\cdot\equiv\cdot$. The
implicit reflexivity and transitivity of state equivalence comply with linear
judgement due to the $(Identity)$ and $(Cut)$ rules.
\end{proof}
\end{lemma}

Theorem~\ref{thm:soundness} states the soundness of our semantics.

\begin{theorem}[Soundness]
  \label{thm:soundness}
  Let $\bbP$ be a program, $CT$ be a constraint theory, and
  $\Sigma=\Sp\cup\Sct\cup\Seq$. Then for arbitrary states $S,T$, we have:
  \[
  	S\mapsto^{*}T \Rightarrow S^L\vdash_\Sigma T^L
  \]
\begin{proof}
Let $S,T$ be states such that $S\mapsto^{r}T$. According to
Def.~\ref{def:chr-op-sem}, there exists a variant of a rule with fresh variables
$r\ @\ H_1\setminus H_2 \Leftrightarrow G\mid B_u \wedge B_b$ and states
$S'=\state{H_1\wedge H_2 \wedge \U; G\wedge\B; \V}$, $T'=\state{B_u \wedge
H_1\wedge \U; B_b\wedge G\wedge\B; \V}$  such that $S'\equiv S$ and $T'\equiv T$.
Consequently, $\Sp$ contains:
\[
  H_1^L\x H_2^L\x G^L \vdash_\Sigma H_1^L\x \exists\by_r.(B_b^L\x B_u^L\x G^L)
\]
From which we prove:
\[
  \exists_{-\V}.H_1^L\x H_2^L\x G^L\x\U\x\B
   \vdash_\Sigma
  \exists_{-\V}.H_1^L \x\exists\by_r.(B_b^L\x B_u^L\x G^L)\x\U\x\B
\]
The local variables $\by_r$ of $r$ are by Def.~\ref{def:chr-op-sem} disjoint
from $vars(H_1,H_2,\U,\B,\V)$. Hence, we have:
\[
  \exists_{-\V}.H_1^L\x H_2^L\x G^L\x\U\x\B
   \vdash_\Sigma
  \exists_{-\V}.H_1^L\x G^L \x B_b^L\x B_u^L\x\U\x\B
\]
This corresponds to $S'^L\vdash_\Sigma T'^L$. Lemma~\ref{lemma:sq-ll} proves that
${S}^L\vdash_\Sigma {T}^L$. As the judgement relation $\vdash$ is transitive and reflexive, we can generalize the relationship
to the reflexive-transitive closure $S\mapsto^{*} T$.
\end{proof}
\end{theorem}

\subsection{State Entailment}
	\label{sec:entailment}

In this section, we define the notion of \emph{entailment}, which we will use to
formulate our theorem of completeness. We present it alongside various properties
that follow from it and that will be used in upcoming sections.

\begin{definition}
\label{def:s_entail}

\emph{State entailment}, written as $\cdot\ent\cdot$, is the smallest
partial-order relation over equivalence classes of CHR states that satisfies the
following conditions:

\begin{enumerate}
	\item \label{cond:sn_wea} \emph{(Weakening of the Built-In Store)} For states $\state{\U;\B;\V},\state{\U;\B';\V}$ with  local variables $\bs,\bs'$ such that
	$CT\models\forall(\exists\bs.\B\rightarrow\exists\bs.\B')$, we have:
	\[
	  [\state{\U;\B;\V}] \ent [\state{\U;\B';\V}]
	\]
	\item \label{cond:sn_omit} \emph{(Omission of Global Variables)}
	\[
	  [\state{\U;\B;\{x\}\cup\V}] \ent [\state{\U;\B;\V}]
	\]

\end{enumerate}
\end{definition}

To simplify notation, we often write $S\ent S'$ instead of $[S]\ent [S']$.
Theorem~\ref{thm:se_crit} gives a decidable criterion for state entailment. The
criterion requires that the global variables of the entailed state are contained
in the global variables of the entailing state. This is never a problem, as we
may choose representatives of the respective equivalence classes that satisfy the
condition.

\begin{theorem}[Criterion for $\ent$]\label{thm:se_crit} Let $S =
\state{\C;\B;\V}, S' = \state{\C';\B';\V'}$ be CHR states with local
variables~$\bl,\bl'$ that have been renamed apart and where $\V'\subseteq\V$.
Then we have:
\[ [S]\ent [S']
	\quad\Leftrightarrow\quad
	CT\models \forall (\B \rightarrow \exists \bl'.((\U \de \U') \wedge\B'))
\]
\begin{proof}
\noindent'$\Rightarrow$': We show that the explicit axioms of entailment, as well as the implicit conditions reflexivity,
anti-symmetry and transitivity comply with the criterion:
\begin{description}
	\item[Def.~\ref{def:s_entail}.\ref{cond:sn_wea}] We assume w.l.o.g. that the strictly local variables of $\state{\U;\B;\V},\state{\U;\B';\V}$ are renamed apart. We observe that
	$(\U\de\U)$ is
	a tautology for any $\U$. Hence, from $CT\models\forall(\exists\bs.\B\rightarrow\exists\bs.\B')$ follows
	$CT\models\forall(\exists\bs.\B\rightarrow\exists\bl'.(\U\de\U)\wedge\B')$, which proves:
	$
		CT\models \forall (\B\rightarrow\exists\bl'.((\U\de\U)\wedge\B'))
	$
	\item[Def.~\ref{def:s_entail}.\ref{cond:sn_omit}] Let $\bl$ be the
	local variables of $\state{\U;\B;\{x\}\cup\V}$. For any $x$ we have:
	$
		CT\models \forall (\B\rightarrow\exists
		x.\exists\bl.((\U\de\U)\wedge\B))
	$
	\item[Reflexivity] Let $\state{\U;\B;\V},\state{\U';\B';\V'}$ be CHR
states
	such that $[\state{\U;\B;\V}]=[\state{\U';\B';\V'}]$, i.e.
	$\state{\U;\B;\V}\equiv \state{\U';\B';\V'}$. Assuming that the local
	variables $\bl, \bl'$ have been named apart, Thm.~\ref{thm:sq_crit}
implies
	$
		CT\models \forall (\B\rightarrow\exists\by'.((\U\de\U')\wedge\B'))
	$.
	\item[Anti-Symmetry] Let $\state{\U;\B;\V},\state{\U';\B';\V'}$ be CHR
states
	with local variables $\bl,\bl'$ such that $CT\models \forall
	(\B\rightarrow\exists\bl'.((\U\de\U')\wedge\B'))$ and  $CT\models \forall
	(\B'\rightarrow\exists\bl.((\U\de\U')\wedge\B))$. By
Thm.~\ref{thm:sq_crit},
	we have that $\state{\U;\B;\V}\equiv \state{\U';\B';\V'}$ and hence
	$[\state{\U;\B;\V}]=[\state{\U';\B';\V'}]$.
	\item[Transitivity] Let
	$\state{\U;\B;\V},\state{\U';\B';\V'},\state{\U'';\B'';\V''}$ be CHR
states
	where the local variables $\bl,\bl',\bl''$have been renamed apart and
such
	that $CT\models \forall (\B\rightarrow\exists\bl'.((\U\de\U')\wedge\B'))$ and
	$CT\models \forall (\B'\rightarrow\exists\bl''.((\U'\de\U'')\wedge\B''))$.
	Therefore, $CT\models \forall (\B\rightarrow\exists\by'.((\U\de\U')
	\wedge\exists\bl''.((\U'\de\U'')\wedge\B'')))$. As the sets of local
variables
	are disjoint, we get $CT\models \forall
(\B\rightarrow\exists\bl'\bl''.((\U\de\U')
	\wedge(\U'\de\U'')\wedge\B''))$ and finally
	\[
		CT\models \forall (\B\rightarrow\exists\bl''.((\U\de\U'')\wedge\B''))
	\]
\end{description}

\noindent'$\Leftarrow$': Let $S = \state{\U;\B;\V}, S' = \state{\U';\B';\V'}$ be
CHR states with local variables~$\by,\by'$ that have been renamed apart and such
that $\V'\subseteq\V$ and
$
	CT\models \forall (\B \rightarrow \exists \by'.((\U \de \U') \wedge
	\B'))
$.
We apply Def.~\ref{def:s_entail}.\ref{cond:sn_wea} to infer:
$
	S\ent\state{\U;(\U\doteq \U')\wedge \B';\V}
$.
By Def.~\ref{def:s_equiv}.\ref{cond:se_subst} and
Def.~\ref{def:s_equiv}.\ref{cond:se_appct}, we get
$
	S\ent\state{\U';\B';\V}
$.
Since $\V'\subseteq\V$, several applications of
Def.~\ref{def:s_entail}.\ref{cond:sn_omit} give us
$
	S\ent\state{\U';\B';\V'}=S'
$.
\end{proof}
\end{theorem}

Corollary~\ref{crl:ent-equiv} is a direct consequence of Theorem~\ref{thm:sq_crit}
and Theorem~\ref{thm:se_crit}. It establishes the relationship between state
equivalence and state entailment.

\begin{corollary}[($\equiv\Leftrightarrow\triangleleft\ent$)]
\label{crl:ent-equiv}
  For arbitrary CHR states $S,T$, state equivalence $S\equiv_e T$ holds \emph{if
  and only if} both $S\ent T$ and $T\ent S$ hold.
\end{corollary}

Lemma~\ref{lemma:exchange} establishes an important relationship between state transition and state entailment.

\begin{lemma}%[Exchange of $\ent$ and $\mapsto$]
\label{lemma:exchange}
  Let $S,U,T$ be CHR states. If $S\ent
  U$ and $U\mapsto^r T$ then there exists a state $V$ such that
  $S\mapsto^r V$ and $V\ent T$.

\begin{proof}
 Let $S=\state{\U;\B;\V}$ and
let $\by_S,\by_U,\by_T$ be the
    local variables of $S,U,T$.
  By definition, $U\mapsto^r T$ implies that there is a variant of a CHR rule
$r\ @\ (H_1\setminus H_2\Leftrightarrow G\mid
	B_b\wedge B_u)$ such that
    $[U]=[\state{H_1\wedge H_2\wedge{\hat\U};G\wedge\hat\B;\hat\V}]$ and
    $[T]=[\state{H_1\wedge
B_u\wedge{\hat\U};G\wedge{B_b}\wedge{\hat\B};\hat\V}]$.

  Now let $V=\state{H_1\wedge
B_u\wedge{\hat\U};G\wedge{B_b}\wedge{\hat\B}\wedge(\U\de(H_1\wedge
H_2\wedge\hat\U)\wedge\B;\hat\V}$.
  From $[S]\ent[U]$ follows by Thm.~\ref{thm:se_crit}:
    $
    	CT\models \forall(
    		\B\rightarrow\exists\by_U.(
    			(\U\de(H_1\wedge H_2\wedge\hat\U))\wedge G\wedge\hat\B
    		)
    	)
    $.
  Assuming w.l.o.g. that $\by_S\cap\by_U=\emptyset$, we can apply
Def.~\ref{def:s_equiv}.\ref{cond:se_appct} to get
    $
    	S \equiv \state{\U;\B\wedge G\wedge\hat\B\wedge(\U\de(H_1\wedge
H_2\wedge\hat\U));\V}
    $
  and then
    $
    	S \equiv \state{(H_1\wedge H_2\wedge\hat\U;\B\wedge
G\wedge\hat\B\wedge(\U\de(H_1\wedge H_2\wedge\hat\U));\V}
    $.
  According to Def.~\ref{def:chr-op-sem}, we have $S\mapsto_r V$. We apply
Def.~\ref{def:s_entail} to show that $V\ent T$.
\end{proof}
\end{lemma}

In anticipation of Section~\ref{sec:pt-completeness}, the following example shows
how the notion of entailment fills the gap between the computability
relation between states and the judgement relation between their
respective linear-logic readings.

\begin{example}
	\label{example:pt-sem-entailment}
In Example~\ref{example:pt-sem-completeness-problem}, we showed that the
following judgement, which does not correspond to any transition in CHR, is
 provable in our sequent calculus system:
\[
	\state{\top;x\doteq3;\{x\}}^L\vdash_\Sigma\state{\top;\top;\emptyset}^L
\]
We observe that the two states are connected by the entailment
relation:
\[
	\state{\top;x\doteq3;\{x\}}\ent\state{\top;\top;\emptyset}
\]
In the following section, we will show that state entailment precisely covers
the discrepance between transitions in a CHR
program and judgements in its corresponding sequent
calculus system as exemplified in
Example~\ref{example:pt-sem-completeness-problem}
\end{example}

\subsection{Completeness of the Axiomatic Semantics}
	\label{sec:pt-completeness}

The notion of merging is an important tool for the proofs in this section. We
define it as follows:

\begin{definition}[($\cdot\diamond\cdot$)]
	\label{def:merging}
  Let $S=\state{\U;\B;\V}, S'=\state{\U';\B';\V}$ be CHR
  states that share the same set of global variables and whose local variables
  are renamed apart. Their \emph{merging} is defined as:
  \[
  	S\diamond S' ::=\state{\U\wedge\U';\B\wedge\B';\V}
  \]
\end{definition}

The following property assures that we can without loss of generality assume
the existence of $S\diamond T$ for any two states $S,T$:

\begin{property}%[Generality of $\cdot\diamond\cdot$]
	\label{prop:generalityofmerging}
For any CHR states $S,T$, there exist states $S',T'$ where $S\equiv S',T\equiv
T'$ such that $S'\diamond T'$ exists.
\begin{proof}[sketch]
Lemma~\ref{lem:se_derived}.\ref{prop:se_rename} allows to rename the local
variables apart, and Def.~\ref{def:s_equiv}.\ref{cond:se_global} allows the
union of their respective sets of global variables.
\end{proof}
\end{property}

Lemma~\ref{lemma:merging_derived} states two properties of merging that will be
used in upcoming proofs:

\begin{lemma}[Properties of $\cdot\diamond\cdot$]
	\label{lemma:merging_derived}

Let $S,S',T$ be CHR states such that both $S\diamond T$ and $S'\diamond T$
exist. The following properties hold:

\begin{enumerate}
  \item \label{prop:mrg:entail}
  $
  	S\ent S' \quad\Rightarrow\quad S\diamond T\ent S'\diamond T
  $
  \item \label{prop:mrg:maps}
  $
  	S\mapsto^r S' \quad\Rightarrow\quad S\diamond T\mapsto^r S'\diamond T
  $
\end{enumerate}

\begin{proof}
Lemma~\ref{lemma:merging_derived}.\ref{prop:mrg:entail}: We assume w.l.o.g. that
the states $S,S',T$ share the same set of global variables. Let
$S=\state{\U;\B;\V},S'=\state{\U';\B';\V},T=\state{\U_T;\B_T;\V}$ with local vars
$\bl,\bl',\bl_T$.
From $S\ent S'$ follows by Thm.~\ref{thm:sq_crit}:
$
	CT\models \forall (\B \rightarrow \exists \bl'.((\U \de \U') \wedge \B'))
$.
As $\U_T \de \U_T$ is a tautology, we get $
	CT\models \forall (\B\wedge\B_T \rightarrow \exists \bl'.\exists\bl_T.((\U
\de \U')
	\wedge (\U_T \de \U_T) \wedge \B' \wedge \B_T))
$
which proves $S\diamond T \ent S'\diamond T$.

Lemma~\ref{lemma:merging_derived}.\ref{prop:mrg:maps}: We assume w.l.o.g. that
the states $S,S',T$ share the same set of global variables. According to
Def.~\ref{def:chr-op-sem}, there exists a variant of a CHR rule $r\ @\ H_1
\setminus H_2 \Leftrightarrow G\mid B_u \wedge B_b$, such that $S\equiv\state{H_1
\wedge H_2 \wedge \U;G\wedge\B;\V}$ and $S'\equiv \state{H_1 \wedge B_c \wedge
\U;G\wedge B_b\wedge\B;\V}$. By Prop.~\ref{prop:generalityofmerging}, there
exists a state $T'=\state{\U';\B';\V}$ such that $T'\equiv T$ whose local
variables are renamed apart from those of $S$ and $T$. By
Def.~\ref{def:chr-op-sem}, we get $S\diamond T\mapsto^r S'\diamond T$.
\end{proof}
\end{lemma}

Lemma~\ref{lemma:proof-structure} sets the stage for the completeness theorem:

\begin{lemma}
\label{lemma:proof-structure}
Let $\pi$ be some cut-reduced proof of a sequent $S^L\vdash_\Sigma T^L$, where
$S,T$ are arbitrary CHR states and $\Sigma=\Sp\cup\Sct\cup\Seq$ for a program
$\bbP$ and a constraint theory $CT$. Any formula $\alpha$ in $\pi$ is either of the form
$\alpha=S_\alpha^L$ where $S_\alpha$ is a CHR state or of the form
$\alpha=c_b(\bt)$ where $c_b(\bt)$ is some built-in constraint.
\begin{proof}
We observe that both the root of $\pi$ and all proper axioms in $\Sigma$ are of
the form $U_1^L\vdash U_2^L$ where $U_1,U_2$ are CHR states. The subformula property hence
guarantees that every formula $\alpha$ in $\pi$ is a subformula of the logical
reading $U^L$ of some CHR state $U$. The general form of such a logical reading
is $U^L=\exists l_1.\ldots.\exists l_n.(c^1_u(\bt_1)\x\ldots\x c^m_u(\bt_m))\x
(!c^1_b(\bt'_1)\x\ldots\x !c^k_b(\bt'_k))$ where $l_1,\ldots,l_n$ are the local
variables of $U$, $c^1_u(\bt_1),\ldots, c^m_u(\bt_m)$ are its user-defined
constraints and $c^1_b(\bt_1),\ldots, c^k_b(\bt_k)$ are its built-in constraints.
We observe that any subformula $\alpha$ of $U^L$ is either of the form
$\alpha=S_\alpha^L$ for some CHR state $S_\alpha$ or of the form
$\alpha=c_b(\bt)$, where $c_b(\bt)$ is a built-in constraint.
\end{proof}
\end{lemma}

The completeness of our semantics is formulated in
Theorem~\ref{thm:completeness}:

\begin{theorem}[Completeness]
\label{thm:completeness}
   Let $S,T$ be CHR states, $\bbP$ be a CHR program, $CT$ be a constraint theory
   and let  $\Sigma=\Sp\cup\Sct\cup\Seq$. If $S^L\vdash_\Sigma T^L$, then there
   exists a state $T'$ such that $S\mapsto^{*} T'$ and $T'\ent T$ in $\bbP$.
\begin{proof}
To preserve of clarity, we will omit the set $\Sigma$ of proper axioms from the
judgement symbol $\vdash_\Sigma$. Throughout the proof, $\cD_n(U,V)$ denotes
the fact that for CHR states $U,V$, there exist states $U_1,\ldots, U_n$ such that:
	\[
		U \mapsto U_1 \ldots \mapsto U_n \ent V
	\]
Consequently, $\cD_0(U,V)$ equals $U\ent V$.

Secondly, we define an operator on formulas analogoous to merging on states: For
any two (possibly empty) sequences of variables $\bx,\by$ and quantifier-free formulas $\alpha,\beta$ let
$\exists\bx.\alpha\Diamonddot\exists\by.\beta::=\exists\bx.\exists\by.\alpha\otimes\beta$.
We observe that for arbitrary CHR states $U,V$ where $U\diamond V$ exists, we
have $U^L\Diamonddot V^L \equiv (U\diamond V)^L$.
In the following, we assume w.l.o.g. that all existentially quantified variables
in the antecedent of a sequent occuring in $\pi$ are renamed apart. Hence, for
every two formulas of the form $U^L,V^L$ occurring in
the antecedent of one sequent in $\pi$, both $U\diamond V$ and $U^L\Diamonddot V^L$
exist.

We introduce a completion function $\eta$, defined by the following table,
where $U$ is a CHR state, $\cb$ is a built-in constraint and
$\Gamma\vdash\alpha$ is a sequent:

	\medskip
	\begin{tabular}{l @{\hspace{1mm} $::=$ \hspace{1mm}} l}
	  $\eta(U^L)$ & $U^L$ \\
	  $\eta(\cb)$ & $!\cb$ \\
	  $\eta(\gamma,\Gamma)$ & $\eta(\gamma)\Diamonddot\eta(\Gamma)$ \\
	  $\eta(\Gamma\vdash\alpha)$ & $\eta(\Gamma)\vdash\eta(\alpha)$
	  	\quad for non-empty $\Gamma$ \\
	  $\eta(\vdash\alpha)$ & $\lone\vdash\eta(\alpha)$ \\
	\end{tabular}

\medskip
For a sequent $\Gamma\vdash\alpha$, we call $\eta(\Gamma\vdash\alpha)$ the
$\eta$-completion of $\Gamma\vdash\alpha$. From Lemma~\ref{lemma:proof-structure}
follows that for every sequent $\Gamma\vdash\alpha$ in $\pi$, its
$\eta$-completion $\eta(\Gamma\vdash\alpha)$ is of the form $U^L\vdash V^L$ for
some CHR states $U,V$. For example,
	\begin{align*}
	\eta(\exists y.c_u(x,y),x\doteq 1 \vdash \exists z.c_u(1,z)) &
		 =\exists y.c_u(x,y)\x \bang x\doteq 1\vdash \exists
		z.c_u(1,z)\\
	&
		= \state{c_u(x,y)\wedge x\doteq
		1;\{x\}}^L\vdash\state{c_u(1,z);\top;\emptyset}^L
	\end{align*}

	We show by induction over the depth of $\pi$ that for every such
	$U^L\vdash V^L$, we have $\cD_n(U,V)$, where $n$ is the number of
	$\Sp$-axioms in the proof of $U^L\vdash V^L$.

	\textbf{Base case:}
	In case the proof of $S^L\vdash T^L$ consists of a single leaf, it is either an
	instance of a $(Identity)$, $(R\lone)$, or $(L\lzero)$, or a proper axiom
	$(\Gamma\vdash\alpha)\in(\Seq\cup\Sct\cup\Sp)$.

	\begin{itemize}
	    \item $(Identity)$, $(R\lone)$, $(L\lzero)$:
	    	\[
	 			\infer[(Identity)]
					{\alpha\vdash\alpha}
					{}
				\quad
				\infer[(R\lone)]
					{\vdash\lone}
					{}
				\quad
				\infer[(L\lzero)]
					{\lzero\vdash\alpha}
					{}
			\]
	    In the case of $(Identity)$, we have $\eta(\alpha\vdash\alpha)=U^L\vdash
	    U^L$ for some CHR state $U^L$. In the case of $(R\lone)$, we have
	    $\eta(\vdash\lone)=U^L\vdash U^L$ for $U^L=\state{\top;\top;\V}$. As the
	    entailment relation is reflexive, we have $\cD_0(U,U)$. In the case of
	    $(L\lzero)$, we have $\eta(0\vdash\alpha)=U^L\vdash V^L$ where $U\equiv S_\bot$.
	    By  Def.~\ref{def:s_equiv}.\ref{cond:se_fail} and
	    Def.~\ref{def:s_entail}.\ref{cond:sn_wea}, we have that $U^L\ent V^L$ and
	    therefore $\cD_0(U,V)$.

		\item For a proper axiom $(\Gamma\vdash\alpha)\in(\Seq\cup\Sct)$ we have
		$\Gamma\vdash\alpha=U^L\vdash V^L$ where $U,V$ are CHR states such
		that $U\ent V$ and therefore $\cD_0(U,V)$.

		\item For a proper axiom $(\Gamma\vdash\alpha)\in\Sp$ we have
		$\Gamma\vdash\alpha=U^L\vdash V^L$ where $U,V$ are CHR states such that
		$U\mapsto \hat V$ and therefore $\cD_1(U,V)$.
	\end{itemize}

	\textbf{Induction step:} We distinguish nine cases according to
	which is the last inference rule applied in the proof. Cut reduction
	implies that it must be one of $(Cut)$, $(L\otimes)$, $(R\otimes)$,
	$(Weakening)$, $(Dereliction)$, $(Contraction)$, $(R!)$, $(L\exists)$,
	and $(R\exists)$.
	\begin{itemize}
	\item $(L\otimes),(Dereliction),(R!)$: For $(Dereliction)$ and $(R!)$,
	the banged formula must be an atomic built-in constraint $\cb$:
	\[
		\infer[(L\x)]
			{\Gamma,\alpha\x \beta\vdash\gamma}
			{\Gamma,\alpha,\beta\vdash\gamma}
		\quad
		\infer[(Dereliction)]
			{\Gamma,!\cb\vdash\beta}
			{\Gamma,\cb\vdash\beta}
		\quad
		\infer[(R!)]
		{\bang\Gamma\vdash !\cb}
		{\bang\Gamma\vdash \cb}
	\]
	Since $\eta(\alpha,\beta)=\eta(\alpha\x\beta)$ and $\eta(!\cb)=\eta(\cb)$, each
	of these rule is invariant to the $\eta$-completion of the sequent, thus
	trivially satisfying the hypothesis.
	\item $(L\lone)$:
	\[
		\infer[(L\lone)]
			{\Gamma,\lone\vdash\alpha}
			{\Gamma \vdash \alpha}
	\]
	We assume that $S_\Gamma=\state{\C_\Gamma,\B_\Gamma,\V_\Gamma}$ and $S_\alpha$
	are CHR states such that $S^L_\Gamma=\eta(\Gamma)$, $S^L_\alpha=\eta(\alpha)$,
	and $\cD_n(S_\Gamma,S_\alpha)$. Then by
	Def.~\ref{def:s_equiv}.\ref{cond:se_appct}, we have
	$\cD_n(S'_\Gamma,S_\alpha)$ where
	$S'_\Gamma=\state{\U_\Gamma,\B_\Gamma\wedge\top,\V_\Gamma}$. As
	${S'}^L_\Gamma=\eta(\Gamma,\lone)$, this proves the hypothesis.
	\item $(Weakening)$: By Lemma~\ref{lemma:proof-structure}, we have that the
	introduced formula is of the form $\bang\cb$.
	\[
		\infer[(Weakening)]
			{\Gamma, !\cb\vdash\beta}
			{\Gamma\vdash\beta}
	\]
	We assume that $S_\Gamma=\state{\U_\Gamma,\B_\Gamma,\V_\Gamma}$ and $S_\beta$ are
	CHR states such that $S^L_\Gamma=\eta(\Gamma)$, $S^L_\beta=\eta(\beta)$ and
	$\cD_n(S_\Gamma,S_\beta)$. Furthermore,
	let $U=\state{\U_\Gamma;\B_\Gamma\wedge \cb;\V_\Gamma}$. Since
	$U^L=\eta(\Gamma,\bang \cb)$ and $U\ent S_\Gamma$,
	Lemma~\ref{lemma:exchange} proves the hypothesis.
	  \item $(Contraction)$: By the subformula property, we have that the
contracted
	  formula is of the form $!\cb$.
	  \[
	  \infer[(Contraction)]{\Gamma,!\cb\vdash\beta}
	      {\Gamma,!\cb,!\cb\vdash\beta}
	  \]
	  Since $\state{\U;\B\wedge \cb;\V}\ent\state{\U;\B\wedge\cb \wedge
\cb;\V}$
	  we prove the hypothesis analogously to $(Weakening)$.

	  \item $(R\otimes)$: The subformula property implies that the joined
	  formulas must be CHR states $U^L$ and $V^L$ without local variables:
	  \[
	  \infer[(R\otimes)]{\Gamma,\Delta\vdash U^L\x V^L}
	      {\Gamma\vdash U^L & \Delta\vdash V^L}
	  \]
	  Let $S_\Gamma,S_\Delta$ be CHR states such that
	  $S^L_\Gamma=\eta(\Gamma)$,
	  $S^L_\Delta=\eta(\Delta)$. The induction hypothesis gives us
	  $\cD_n(S_\Gamma, U)$ and $\cD_m(S_\Delta, V)$ for some $n,m$. By
	  Lemma~\ref{lemma:merging_derived}.\ref{prop:mrg:entail} and
	  Lemma~\ref{lemma:merging_derived}.\ref{prop:mrg:maps} we have
	  $\cD_n(S_\Gamma\diamond S_\Delta, U\diamond S_\Delta)$ and
	  $\cD_m(U\diamond S_\Delta, U\diamond V)$. By
Lemma~\ref{lemma:exchange},
	  we get $\cD_{n+m}(S_\Gamma\diamond S_\Delta, U\diamond V)$.

	  \item $(Cut)$: Since $\pi$ is a cut-reduced proof and all axioms are of the
	  form $U_1^L\vdash U_2^L$, the eliminated formula must be the logical reading of
	  a CHR state $U$:
	  \[
	  \infer[(Cut)]{\Gamma,\Delta\vdash \beta}
	      {\Gamma\vdash U^L & U^L, \Delta\vdash \beta}
	  \]
	  Let $S_\Gamma,S_\Delta,S_\beta$ be CHR states such that
$S^L_\Gamma=\eta(\Gamma)$,
	  $S^L_\Delta=\eta(\Delta)$, and $S^L_\beta=\eta(\beta)$. The induction
	  hypothesis gives us $\cD_n(S_\Gamma, U)$ and $\cD_m(U\diamond
S_\Delta,
	  S_\beta)$. Applying
Lemma~\ref{lemma:merging_derived}, we get
	  $\cD_n(S_\Gamma\diamond S_\Delta, U\diamond S_\Delta)$. By
Lemma~\ref{lemma:exchange},
	  we get $\cD_{n+m}(S_\Gamma\diamond S_\Delta, S_\beta)$ which proves
	  the hypothesis.

	  \item $(L\exists)$: In the preconditional sequent, the quantified
	  variable $x$ is by definition replaced by a
	  fresh constant $a$ that does not occur in $\Gamma$, $\alpha$, or
	  $\beta$:
	  \[
	  \infer[(L\exists)]{\Gamma,\exists
	  x.\alpha\vdash\beta} {\Gamma,\alpha\subxa\vdash\beta}
	  \]
	  Let $U=\state{\U\subxa;\B\subxa;\V\cup\{a\}}$ and $S_\beta$ be CHR
	  states such that
	  $U^L=\eta(\Gamma,\alpha\subxa)$, $S_\beta^L=\eta(\beta)$, and
	  $x\not\in\V$. The
	  definition of state equivalence gives us $U\equiv \state{\U;\B\wedge
	  x=a;\V\cup\{a\}}$. Furthermore, we have $\eta(\Gamma,\exists
	  x.\alpha)=\state{\U,\B,\V}^L$.
	  By the induction hypothesis, we have states
	  $U_1,\ldots,U_n$
	  such that $U\mapsto^{r_1} U_1\mapsto^{r_2}\ldots\mapsto^{r_n} U_n\ent
S_\beta$ where
	  $U_i=\state{\U_i;\B_i\wedge x=a;\V\cup\{a\}}$ for $\inintv{i}{n}$.
	Neither the binding $x=a$ nor the set of global variables affect rule
	applicability. Hence, we can construct an analogous derivation
	$\state{\U;\B;\V}\mapsto^{r_1} U'_1\mapsto^{r_2}\ldots\mapsto^{r_n} U'_n$ where
	$U'_i=\state{\U_i;\B_i;\V}$ for $\inintv{i}{n}$. Since $U_n\ent S_\beta$ and $a$
	must not occur in $\beta$, we also have have $\state{\U_n;\B_n;\V}\ent S_\beta$.
	  Therefore, we have $\cD_n(\state{\U;\B;\V},S_\beta)$. As
	  $\eta(\Gamma,\exists x.\alpha)=\state{\U;\B;\V}^L$, this proves the
hypothesis.

	  \item $(R\exists)$: By definition, the quantified variable $x$
	  substitutes an
	  arbitrary term $t$.
		\[
			\infer[(R\exists)]{\Gamma\vdash\exists x.\beta}
	            {\Gamma\vdash\beta[x/t]}
	  	\]
	  Let $S_\Gamma,U,V$ be CHR states such that $S_\Gamma^L=\eta(\Gamma)$,
	  $U^L=\eta(\beta\subxt)$, and $V^L=\eta(\exists x.\beta)$. By the
	  induction
	  hypothesis we have $\cD_n(S_\Gamma,U)$ for some $n$. Let
	  $V=\state{\U;\B;\V}$ and $U=\state{\U\subxt;\B\subxt;\{x\}\cup\V}$.
	  We have
	$U
		\equiv \state{\U;\xet\wedge\B;\{x\}\cup\V}
		\ent \state{\U;\xet\wedge\B;\V}
		\ent \state{\U;\B;\V}
		\equiv V
	$,
	and therefore, $\cD_n(S_\Gamma,V)$.
	\end{itemize}
	Finally, we have $\cD_N(S,T)$ for some $N$, i.e. there exist states
	$S_1,\ldots, S_N$ such that:
	\[
		S \mapsto S_1 \mapsto \ldots \mapsto S_N \ent T
	\]
	It follows that for $T'=S_N$, we have $S\mapsto^{*}\hT$ and $\hT\ent T$.
\end{proof}
\end{theorem}

Lemma~\ref{lemma:judge-entail} states that when excluding the proper axioms in
$\Sp$, logical judgement implies state entailment:

\begin{lemma}[($\vdash\,\Leftrightarrow\,\ent$)]
\label{lemma:judge-entail}
   For arbitrary CHR states $S,T$, entailment $S\ent T$ holds \emph{if and only
   if} the judgement $S^L\vdash_\Sigma T^L$ is provable for
   $\Sigma=\Sct\cup\Seq$. \begin{proof}[sketch] \noindent $'\Leftarrow'$: We
   apply Thm.~\ref{thm:completeness} to the empty program $\bbP=\emptyset$.

\noindent $'\Rightarrow'$: We proof that all conditions in
Def.~\ref{def:s_entail} comply with the judgement relation $\vdash$: For
Def.~\ref{def:s_entail}.\ref{cond:sn_wea}, $CT\models\forall(\B\rightarrow\B')$
implies that $\Sct$ contains an axiom $\B\vdash\B'$. Hence, we can prove
$\exists_{-\V}.\U\wedge\B\vdash\exists_{-\V}.\U\wedge\B'$. For
Def.~\ref{def:s_entail}.\ref{cond:sn_omit}, it is valid since $S^L\vdash\exists
x.S^L$ holds for any $S^L$. Concerning the implicit conditions of a partial order
relation, reflexivity and anti-symmetry hold for the judgement relation $\vdash$
as well and anti-symmetry is a natural consequence of Def.~\ref{def:ll-equiv}.
\end{proof}
\end{lemma}

Theorem~\ref{thm:equiv-ll} defines the relationship between state equivalence
and the linear-logic semantics. It is a direct consequence of
Corollary~\ref{crl:ent-equiv} and
Lemma~\ref{lemma:judge-entail} and therefore goes without proof:

\begin{theorem}[($\equiv\Leftrightarrow\dashv\vdash$)]
\label{thm:equiv-ll}
Let $CT$ be a constraint theory and $\Sigma=\Sct\cup\Seq$.
For arbitrary CHR states $S,T$, we have:
\[
S\equiv T \Leftrightarrow S^L\dashv\vdash_\Sigma T^L
\]
\end{theorem}

The following example illustrates the completeness theorem:

\begin{example}
We consider the partial-order program $\bbP$ given in Example~\ref{example:leq} and a minimal
constraint
theory $CT$. For $\Sigma=\Sp\cup\Sct\cup\Seq$, we have
\[
a\leq b\otimes b\leq c\otimes c\leq a\vdash_\Sigma \bang a\doteq b
\]
which equals:
\[
\state{a\leq b \wedge b\leq c \wedge c\leq a;\top;\{a,b,c\}}^L\vdash_\Sigma
\state{\top;a\doteq b;\{a,b\}}^L
\]
This corresponds to:
\[
\state{a\leq b \wedge b\leq c \wedge c\leq a;\top;\{a,b,c\}}
\mapsto^*
\state{\top;a\doteq b \wedge a\doteq c;\{a,b,c\}}
\ent
\state{\top;a\doteq b;\{a,b\}}
\]
\end{example}

\subsection{Encoding Programs and Constraint Theories}
\label{sec:enc-sem}

In the axiomatic linear-logic semantics presented in Sect.~\ref{sec:pt-sem} to
Sect.~\ref{sec:pt-completeness}, only states are represented in logical
judgements. Both programs and constraint theories disappear into the proper
axioms of a sequent calculus system and hence are not objects of logical
reasoning.

In this section, we show how to encode programs and constraint theories into
logical judgements, enabling us to reason directly about them as well. In
Sect.~\ref{sec:app:comparison}, we will use this encoding to decide operational
equivalence of programs. As a further benefit, a complete encoding of programs
and constraint theories assures the existence of cut-free proofs for the
respective judgements and ensure compatibility with established methods for
automated proof search methods relying on this property.

As usual, $(\cdot)^L$
stands for translation into linear logic.

\paragraph*{Encoding of Constraint Theories}
The constraint theory $CT$ itself is encoded according to the translation quoted
in Def.~\ref{def:negri}. Furthermore, for every n-ary
user-defined constraint symbol $c_u$ and every $\inintv{j}{n}$, we add the
following formula to the translation of the theory, where $x_1,\ldots,x_n$ and
$y$ are variables:
\[
!\forall(c_u(x_1,...,x_j,...,x_n) \x !(x_j\doteq y)
\multimap c_u(x_1,...,y,...,x_n) \x !(x_j\doteq y)
\]
We obtain the following encoding of constraint theories:
\begin{definition}[($CT^L$)]
\label{def:ct-l}
Let $CT$ be a constraint theory. Its linear-logic reading $CT^L$ is given
as:
\[
CT^L ::=
CT^* \cup
\left(
\bigcup_{c_u/n}\bigcup_{j\doteq 1}^n
\bang\forall(c_u(...,x_j,...) \x !(x_j= y)
\multimap c_u(...,y,...) \x !(x_j= y)
\right)
\]
\end{definition}
\paragraph*{Encoding of $\Sp$} The translation of CHR rules follows the same
lines as the encoding of the $CT$ axioms:
\begin{definition}[($R^L,\bbP^L$)]
\label{def:rp-l}
\begin{enumerate}
\item \label{def:rp-l:r-l} Let $R = r\ @\ H_1\setminus H_2\Leftrightarrow G\mid B_u \wedge
B_b$ be a CHR rule with local variables $\by_r$. Then its linear-logic reading $R^L$ is defined as:
\[
R^L ::= !\forall(H_1^L\x H_2^L\x G^L \multimap
H_1^L\x \exists\by_r.(B_u^L\x B_b^L\x G^L))
\]
\item \label{def:rp-l:p-l} Let $\bbP = \{R_1,\ldots,R_n\}$ be a CHR program. Then
its linear-logic
reading $\bbP^L$ is defined as:
\[
\bbP^L ::= \bigcup_{R\in \bbP} R^L
\]
\end{enumerate}
\end{definition}

For the encoding semantics, the following soundness and completeness theorem
holds:

\begin{theorem}[Soundness and Completeness]
  \label{theorem:embed_soundness_completeness}
  Let $S,T$ be CHR states. There exists a state $U$ such that
  \[
    S\mapsto^{*}U\textrm{ and }U\ent T
  \]
  in a program $\bbP$ and a constraint theory $CT$ \emph{if and only if}
  \[
    \bbP^L,CT^L \vdash \forall(S^L\lp T^L)
  \]
\begin{proof}
  We prove Thm. \ref{theorem:embed_soundness_completeness} by showing that any
  proof tree in the axiomatic semantics can be transformed into a proof tree in
  the encoding semantics and vice versa. To ensure of clarity, we will omit
  the set of proper axioms from the judgement symbol.

  \paragraph*{Axiomatic to encoding:} We assume a proof $\pi$ of a sequent
  $S^L\vdash T^L$ in the axiomatic semantics. We replace
  every axiom $\exists\bx.\B^L\vdash\exists\bx'.{\B'}^L$ in $\Sct$ by a sub-tree proving
  $CT^L,\exists\bx.\B^L\vdash\exists\bx'.{\B'}^L$. The same is done for every
  equality axiom in $\Seq$. Similarly, every axiom $H^L_1\x H^L_2\x G^L\vdash
  H^L_1\x\exists_{-\by_r}.(B^L_u\x B^L_b\x G^L)$ in $\Sp$ is replaced with a
  sub-tree proving $\bbP^L, H^L_1\x H^L_2\x G^L\vdash
  H^L_1\x\exists_{-\by_r}.(B^L_u\x B^L_b\x G^L)$. We propagate the thus
  introduced instances of $CT^L$ and $\bbP^L$ throughout the proof tree, thus
  producing a proof $\pi'$ of
  \[
  	CT^L,\ldots,CT^L,\bbP^L,\ldots,\bbP^L,S^L\vdash T^L
  \]
  We insert $\pi'$ into:
  \[
  	\small
  	\infer[(R\forall)]
  	{
	  	CT^L,\bbP^L\vdash \forall(S^L\lp T^L)
  	}
  	{
  		\infer[(R\lp)]
  		{
  			CT^L,\bbP^L\vdash S^L\lp T^L
  		}
  		{
  			\infer[(Contraction)^*]
  			{
  				CT^L,\bbP^L,S^L\vdash T^L
  			}
  			{
  				\pi'
  			}
  		}
  	}
  \]

  \paragraph*{Encoding to axiomatic:} Let $\bigotimes$ stand for element-wise
  multiplicative conjunction of a set and let $\pi$ be a proof of a sequent
  $CT^L,\bbP^L\vdash \forall(S^L\lp T^L)$ in the encoding semantics.

  For every $!\forall(\exists\bx.{\B}^L\lp \exists\bx'.{\B'}^L)\in CT^L$, we
  have $\vdash_\Sigma !\forall(\exists\bx.{\B}^L\lp \exists\bx'.{\B'}^L)$ where
  $\Sigma=\Sct\cup\Seq$. Hence, there exists a proof $\pi_{CT}$ of
  $\vdash_\Sigma \bigotimes CT^L$. Similarly, there exists a proof $\pi_{\bbP}$
  of $\vdash_{\Sp} \bigotimes \bbP^L$.

  \[
  	\small
	\infer[(Cut)]
	{
		S^L \vdash T^L
 	}
 	{
		\infer[(Cut)]
		{
			\vdash \forall(S^L\lp T^L)
		}
		{
			\pi_{\bbP}
		&
			\infer[(Cut)]
			{
				\bbP^L \vdash \forall(S^L\lp T^L)
			}
			{
				\pi_{CT}
			&
				\infer[(L\otimes)^{*}]
				{
					\bigotimes CT^L, \bigotimes \bbP^L \vdash \forall(S^L\lp T^L)
				}
				{
					\pi
				}
			}
		}
	&
		\infer[(L\forall)^{*}]
		{
			\forall(S^L\lp T^L), S^L\vdash T^L
		}
		{
			\infer[(L\lp)]
			{
				S^L, S^L\lp T^L\vdash T^L
			}
			{
				\infer[(Identity)]
				{
					S^L \vdash S^L
				}
				{
				}
			&
				\infer[(Identity)]
				{
					T^L \vdash T^L
				}
				{
				}
			}
		}
 	}
  \]

As we can transform the respective proof tree from the axiomatic to the
encoding semantics and vice versa, the two representations are equivalent.
\end{proof}
\end{theorem}

%%%%%%%%%%%%%%%%%%%%%%%%%%%%%%%%%%%%%%%%%%%%%%%%%%%%%%%%%%%%%%%%%%%%%%%%%%%
%%%%%%%%%%%% CHRv %%%%%%%%%%%%%%%%%%%%%%%%%%%%%%%%%%%%%%%%%%%%%%%%%%%%%%%%%
%%%%%%%%%%%%%%%%%%%%%%%%%%%%%%%%%%%%%%%%%%%%%%%%%%%%%%%%%%%%%%%%%%%%%%%%%%%

\section{A Linear-Logic Semantics for CHR$^\vee$}
  \label{sec:chrv}

In this section, we extend our linear logic semantics to \emph{CHR with
Disjunction} (CHR$^\vee$), a common extension of CHR. To avoid ambiguity, we
will henceforth use the term \emph{pure CHR} to refer to the regular segment of
CHR without disjunction.

We will firstly recall the syntax and semantics of CHR$^\vee$ in
Sect.~\ref{sec:vee-intro}. Then we define an equivalence-based formalization of
its operational semantics in Sect.~\ref{sec:vee-oesq}, analogous to $\oesq$ for
pure CHR. In Sect.~\ref{sec:vee-extend}, we apply this equivalence-based
formalization to define a linear-logic semantics for CHR$^\vee$ and proof its soundness and
completeness. In Sect.~\ref{sec:vee-congruence}, we show that in the case of
CHR$^\vee$, the linear-logic semantics has less desirable properties than for
pure CHR: Concretely, linear-logic based reasoning over CHR$^\vee$
programs produces in general less precise results than over CHR programs. We
then introduce the well-behavedness properties of \emph{compactness} and
\emph{analyticness} which amend this limitation.

\subsection{Introduction to CHR$^\vee$}
\label{sec:vee-intro}

CHR$^\vee$ has a richer syntax than pure CHR: The definition of goals
is extended by the disjunction operator $\vee$. Alluding to its operational meaning,
we may also refer to $\vee$ as the \emph{split operator}. We also introduce the
notion of \emph{configuration}, which can be read as a disjunction of CHR states,
and we extend the definition of goal equivalence to account for
distributivity.

\begin{definition}[Goals, States, Configurations]
	\label{chrv-state-conf}
We adapt the definitions of goal and state, and we define configuration as
follows:

\medskip
\begin{tabular}{l @{\quad} r @{\,::=\,} l}
	Built-in constraint: &
		$\B$ &
			$\top \mid \cb \mid \B\wedge\B'$\\
	User-defined constraint: &
		$\U$ &
			$\top \mid \cu \mid \U\wedge\U'$\\
	\chrv\ goal: &
		$\G$ &
			$\top \mid \cu \mid \cb \mid \G \wedge \G' \mid \G\vee\G'$\\
	\chrv\ state: &
		$S$  &
			$\state{\G;\V}$ \\
	Configuration: &
		$\bS$ &
			$\varepsilon \mid S \mid S\vee\bS$
\end{tabular}

\medskip For any two goals $\G,\G'$, goal equivalence $\G\equiv_{G}\G'$ denotes
equivalence between goals with respect to \emph{associativity} and
\emph{commutativity} of $\wedge$, the \emph{neutrality} of $\top$ with respect to
$\wedge$, and the \emph{distributivity} of $\wedge$ over $\vee$. $\varepsilon$
stands for the empty configuration, which is operationally equivalent to a
failed state $S_\bot$.

A goal which does not contain disjunctions is called \emph{flat}. A
state $\state{\G;\V}$ where $\G$ is flat is also called flat. A
configuration $\bS$ is called flat if it is empty or consists only of flat
states.
\end{definition}

Allowing $\wedge$ to distribute over $\vee$ guarantees that every goal is
equivalent to its disjunctive normal form (DNF). We do not allow the opposite
law of distributivity. For example, we have $\G_1\wedge(\G_2\vee\G_3)\equiv_G
(\G_1\wedge\G_2)\vee(\G_1\wedge\G_3)$ but $\G_1\vee(\G_2\wedge\G_3)\not\equiv_G
(\G_1\vee\G_2)\wedge(\G_1\vee\G_3)$. Thus any finite goal has only a finite
number of equivalent representations.

In CHR$^\vee$, we use the same definition for state equivalence as in pure CHR.
However, as the
definition of goal equivalence is extended, this implicitly carries over to
state equivalence. For example:
$\state{\G_1\wedge(\G_2\vee\G_3);\V}\equiv\state{(\G_1\wedge\G_2)\vee(\G_1\wedge\G_3);\V}$.

As in goals, CHR$^\vee$ allows disjunctions in rule bodies. The clear
seperation between user-defined constraints and built-in constraints in the
rule body no longer applies. This is reflected in the following definition:

\begin{definition}[\chrv\ Rules]
	\label{def:chrv-rule}
A \chrv\ rule is of the form
\[
	r\ @\ H_1\setminus H_2\Leftrightarrow G \mid B
\]
The \emph{kept head} $H_1$ and the \emph{removed head} $H_2$ are
user-defined constraints. The guard $G$ is a built-in
constraint. The rule body $B$ is a \chrv\ goal. $r$ serves as an identifier for
the rule and may be omitted along with the $@$. An empty guard may be omitted
along with the $\mid$.
\end{definition}

We observe that restricting \chrv\ to the segment without disjunction restores
pure CHR. Hence, pure CHR is a subset of \chrv. The operational semantics of
CHR$^\vee$ has originally been defined in \cite{DBLP:conf/fqas/AbdennadherS98}.
An additional transition rule called \textbf{Split} resolves disjunctions by
branching the computation. Adjusted to our syntax, we express that transition rule as
follows:
\[
	\textbf{Split: }\quad
	\state{ \G_1\vee \G_2 ; \V}
		\mapsto^{sp}
	\state{ \G_1 ; \V}
		\vee
	\state{ \G_2 ; \V}
\]

We can straightforwardly adapt the operational semantics $\oesq$ to the syntax of
CHR$^\vee$. Adding one rule to handle equivalence transformations of states and
two more rules to handle composition of configurations gives us the following
operational semantics for CHR$^\vee$:

\begin{definition}[Operational Semantics of CHR$^\vee$]
CHR$^\vee$ is a state transition system over configurations
defined by the following transition rules, where $(r\ @\ H_1 \setminus H_2
\Leftrightarrow G\mid B)$ is a variant of a CHR$^\vee$ rule whose local
variables $\by_r$ are renamed apart from any variable occurring in $vars(H_1, H_2,\G,\V)$:
\[
	\textbf{Apply: }\quad
	\frac{r\ @\ H_1 \setminus H_2 \Leftrightarrow G\mid B
		\quad\quad\quad
	CT\models \exists(G\wedge\B)}
{
	\state{H_1 \wedge H_2 \wedge G\wedge\G;\V}
		\mapsto^r
	\state{H_1 \wedge G\wedge B\wedge\G;\V}
}
\]
\[
	\textbf{Split: }\quad
	\frac{
		\phantom{\top}
	}
	{
		\state{ \G_1\vee \G_2 ; \V}
			\mapsto^{sp}
		\state{ \G_1 ; \V}
			\vee
		\state{ \G_2 ; \V}
	}
\]
\[
	\textbf{StateEquiv: }\quad
	\frac{
		S'\equiv S
			\quad\quad\quad
		S\mapsto^{(r/sp)} T
			\quad\quad\quad
		T\equiv T'
	}
	{
		S'
			\mapsto^{(r/sp)}
		T'
	}
\]
\[
	\textbf{CompLeft: }\quad
	\frac{
		\bS\mapsto^{(r/sp)} \bS'
	}{
		\bS\vee\bT\mapsto^{(r/sp)} \bS'\vee\bT
	}
	\quad\quad\quad
	\textbf{CompRight: }\quad
	\frac{
		\bS\mapsto^{(r/sp)} \bS'
	}{
		\bT\vee\bS\mapsto^{(r/sp)} \bT\vee\bS'
	}
\]\vspace{1mm}

If the applied rule is obvious from the context or irrelevant, we write
transition simply as $\mapsto$. We denote its reflexive-transitive closure
as $\mapsto^{*}$.
\end{definition}

\noindent The following example shows a possible computation in \chrv:

\begin{example}
	\label{example:albatross}
Consider the following CHR$^\vee$ program:
\[
  \begin{array}{lcl}
    r1\ @\ bird & \Leftrightarrow & albatross \vee penguin\\
    r2\ @\ penguin\wedge flies & \Leftrightarrow & \bot\\
  \end{array}
\]
Running this program with the initial state $\state{bird\wedge
flies;\emptyset}$ produces the following fixed-point computation:
\[
  \begin{array}{ l l}
    &	[\state{bird\wedge flies;\emptyset}]\\
    \mapsto^{r1} & [\state{(albatross\wedge flies) \vee (penguin\wedge flies;\emptyset}]\\
    \mapsto^{sp} &
    [\state{albatross\wedge flies;\emptyset}] \vee
        [\state{penguin\wedge flies;\emptyset}]\\
    \mapsto^{r2} & [\state{albatross\wedge flies}] \vee
        [\state{\bot;\emptyset}]\\
  \end{array}
\]
The first transition step is justified by the $\textbf{Apply}$ as well as the
$\textbf{StateEquiv}$ transition rule. The last transition step is justified by $\textbf{Apply}$ and
$\textbf{CompLeft}$.
\end{example}

\subsection{An Equivalence-Based Operational Semantics for CHR$^\vee$}
\label{sec:vee-oesq}

While the operational semantics presented in Sect.~\ref{sec:vee-intro}
precisely formalizes the execution of a CHR$^\vee$ program, it is of limited use
for program analysis. For example, we would intuitively assume that two
configurations should be considered equivalent if they differ only in the order of
their member states.

In this section, we propose a notion of equivalence of configurations, we show
its compliance with rule application and we propose a formalization of the
operational semantics based on equivalence classes of configurations.

\begin{definition}[Equivalence of Configurations]
\label{def:vee-config-equiv}
Equivalence of configurations, denoted as $\cdot\equiv_\vee\cdot$, is the
smallest equivalence relation over configurations satisfying all of the following
properties:
  \begin{enumerate}
      \item \label{cond:vce_ac}
	\emph{Associativity and Commutativity:}
      \[
        \bS \vee \bT
          \equiv_\vee
        \bT \vee \bS
	\quad\textrm{ and }\quad
        (\bS \vee \bT) \vee \bU
          \equiv_\vee
        \bS \vee (\bT \vee \bU)
      \]
      \item \label{cond:vce_steq}
	\emph{State Equivalence}
      \[
        S\equiv_e S'
		\quad\Rightarrow\quad
        S\vee\bT\equv S'\vee\bT
      \]
      \item \label{cond:vce_fail}
	\emph{Neutrality of Failed States:}
      \[
        S_{\bot}\vee\bT
          \equiv_\vee
        \bT
      \]
      \item \label{cond:vce_split}
	\emph{Split:}
      \[
        [\state{\G_1\vee \G_2;\V}] \vee \bT
          \equiv_\vee
        [\state{\G_1;\V}] \vee [\state{\G_2;\V}] \vee \bT
      \]
  \end{enumerate}
\end{definition}

Compliance of configuration equivalence with rule application is formalized as
follows:

\begin{property}[Compliance with Rule Application] Let
$\bS,\bS',\bT$ be arbitrary configurations such that $\bS\equv\bS'$ and
$\bS\mapsto^{*}\bT$. Then there exists a $\bT'$ such that $\bT\equiv\bT'$ and
$\bS'\mapsto\bT'$. \begin{proof}[sketch] Element states of a configuration
are handled independently of each other, making associativity and commutativity
idempotent to rule application. Equivalence transformation of states complies due
to the \textbf{StateEquiv} rule. Failed states do not allow rule
application. Any application of the
\emph{Split} axiom hindering rule application can be reversed by application of
the \textbf{Split} transition.
\end{proof}
\end{property}

The compliance property allows us to define an operational semantics based on
equivalence classes of configurations using only a single transition rule. In
analogy to the equivalence-based semantics $\oesq$ for pure CHR, we will refer
to this operational semantics as $\oesqv$.

\begin{definition}[Transition System of $\oesqv$]
\label{def:vee-tr-system}
CHR is a state transition system over equivalence classes of configurations. It
is defined by the following transition rule, where $(r\ @\ H_1 \setminus H_2
\Leftrightarrow G\mid B)$ is a variant of a CHR rule whose local variables
$\by_r$ are renamed apart from any variable occurring in $vars(H_1,H_2,\G,\V)$:

\medskip
\[
\frac{
	r\ @\ H_1 \setminus H_2 \Leftrightarrow G\mid B
		\quad\quad\quad
	CT\models \exists(G\wedge\B)
}
{
	[\state{H_1 \wedge H_2 \wedge G\wedge\G;\V}\vee\bT]
		\der^r
	[\state{H_1 \wedge G\wedge B\wedge \G;\V}\vee\bT]
}
\]

\medskip
If the applied rule is obvious from the context or irrelevant, we write
transition simply as $\mapsto$. We denote its reflexive-transitive closure
as $\mapsto^{*}$.
\end{definition}

Analogously to pure CHR, we define a notion of confluence:

\begin{definition}[Confluence]
\label{def:vee-confluence}
A CHR$^\vee$ program $\bbP$ is called \emph{confluent}, if for arbitrary
configurations $\bS,\bT,\bU$ such that $[\bS]\mapsto^{*}[\bT]$ and
$[\bS]\mapsto^{*}[\bU]$, there exists a configuration $\bV$ such that
$[\bT]\mapsto^{*}[\bV]$ and $[\bU]\mapsto^{*}[\bV]$.
\end{definition}

Furthermore, we define three sets of observables based on equivalence classes of
configurations:

\begin{definition}[Observables]
	\label{def:observables-vee}
Let $S$ be a CHR state, $\bbP$ be a program, and $CT$ be a constraint theory. We
distinguish the following sets of observables:

\medskip
\begin{tabular}{l @{\,} r @{\,$::=$} l}
Computable config.: &
	$\bcC_{\bbP,CT}(S)$ &
		$\{[\bT]\mid [S]\mapsto^*[\bT]\}$\\
Answer: &
	$\bcA_{\bbP,CT}(S)$ &
		$\{[\bT]\mid [S]\mapsto^*[\bT]\not\mapsto\}$\\
Data-sufficient answer: &
	$\bcS_{\bbP,CT}(S)$ &
		$\{[\state{\top;\B_1;\V_1}\vee\ldots\vee\state{\top;\B_n;\V_n}]\mid$ \\
\multicolumn{2}{c}{} & \quad
		$[S]\mapsto^*[\state{\top;\B_1;\V_1}\vee\ldots\vee\state{\top;\B_n;\V_n}]\}$
\end{tabular}

\medskip
Note that the parameters for all three sets are states rather than
configurations, as we assume that every computation starts from a singular
state. For all three sets, if the constraint theory $CT$ is clear from the
context or not important, we may omit it from the respective identifier.
\end{definition}

Analogously to Property~\ref{prop:op-equiv-hierarchy}, we have a hierarchy of
observables:

\begin{property}[Hierarchy of Observables]
For any state $S$, program $\bbP$ and constraint theory $CT$, we have:
\[
	\bcS_{\bbP,CT}(S) \subseteq \bcA_{\bbP,CT}(S) \subseteq \bcC_{\bbP,CT}(S)
\]
\end{property}

The following example illustrates our definitions:

\begin{example}
\label{example:albatross_vee}
We recur to the program from Example~\ref{example:albatross}.
\[
  \begin{array}{lcl}
    bird & \Leftrightarrow & albatross \vee penguin\\
    penguin\wedge flies & \Leftrightarrow & \bot\\
  \end{array}
\]
Using $\oesqv$, we can construct the following derivation starting from the
initial state $S_0=\state{bird\wedge flies;\emptyset}$:
\[
  \begin{array}{ l l}
    &	[\state{bird\wedge flies;\emptyset}]\\
    \mapsto_\vee & [\state{(albatross\vee penguin)\wedge flies;\emptyset}]\\
    = &
    [\state{albatross\wedge flies;\emptyset} \vee
        \state{penguin\wedge flies;\emptyset}]\\
    \mapsto_\vee & [\state{albatross\wedge flies} \vee
        \state{\bot;\emptyset}]\\
    =  &[\state{albatross\wedge flies}]\\
  \end{array}
\]

In comparison with Example~\ref{example:albatross}, we now obtain our result with
one less transition. More importantly, our transition system consists of only one
transition rule now. The equivalence relation over configurations allows us to
omit the failed state from the final configuration, producing a more elegant
representation of the answer.

With respect to the observables, we have $\bcC_{\bbP}(S_0)=
\{[S_0],[\state{(albatross\vee penguin)\wedge
flies;\emptyset}],[\state{albatross\wedge flies}]\}$, $\bcA_{\bbP}(S_0)=
\{[\state{albatross\wedge flies}]\}$, and $\bcS_{\bbP}(S_0)= \emptyset$.
\end{example}

\subsection{Extending the Linear-Logic Semantics to \chrv}
\label{sec:vee-extend}

In this section, we develop a linear-logic semantics for CHR$^\vee$, based on
the equivalence-based operational semantics $\oesqv$.

\subsubsection{Definition of the Semantics}

Since pure CHR is completely contained in CHR$^\vee$ and represents a
significant subset thereof, it stands to reason that the linear logic semantics
for pure CHR should be preserved for that segment. Hence, a large part of the
semantics carries directly over to CHR$^\vee$. Now consider a pure CHR program
$\bbP_1$ of the following form:
\begin{align*}
  r_1\ @\ H \Leftrightarrow~& G \mid B_1 \\
  r_2\ @\ H \Leftrightarrow~& G \mid B_2
\end{align*}
The logical reading of this program in the encoding semantics is:
\[
\bbP_1^L =
\bang\forall\left(H^L\lp G^L\lp \exists\by_1.B_1^L\right) \otimes
\bang\forall\left(H^L\lp G^L\lp \exists\by_2.B_2^L\right)
\]
This is logically equivalent to:
\[
	\bang\forall\left(H^L\x G^L\lp
    (\exists\by_1.B_1^L) \with
    (\exists\by_2.B_2^L) \right) \dashv\vdash \bbP_1
\]

We gain the insight that don't-care non-determinism in CHR is already
\emph{implicitly} mapped to additive conjunction $\&$ in linear logic.

Mapping the split connective $\vee$ to multiplicative disjunction $\oplus$ is
an obvious choice, as: (1) $\x$ distributes over $\oplus$, (2)
absurdity 0 -- representing failed states -- is neutral with respect to
$\oplus$, and (3) $\oplus$ complements $\&$, which represents committed choice.
Hence we preserve the clear distinction between the two types of
non-determinism. We furthermore adapt the translations of states and programs to
the syntax of CHR$^\vee$, thus obtaining the semantics given in
Fig.~\ref{fig:chrv-axiomatic-semantics}

\begin{figure}
	\begin{center}
	\fbox{
	\begin{tabular}{l @{\quad} r @{$\,::=\,$} l}
	\textrm{Atomic built-in constraints:} & $\cb^L$ & $!\cb$ \\
	\textrm{Atomic user-defined constraints:} & $\cu^L$ & $\cu$ \\
	\textrm{Falsity:} & $\bot^L$ & $\lzero$ \\
	\textrm{Empty constraint/goal:} & $\top^L$ & $\lone$ \\
	\textrm{Constraints/goals:} &
		$(\G_1\wedge\G_2)^L$ & $\G^L_1\otimes\G^L_2$\\
	\textrm{Disjunction within goals:} &
		$(\G_1 \vee \G_2)^L$ & $\G_1^L\oplus \G_2^L$ \\
	\textrm{States:} & $\state{\G;\V}^L$ & $\exists_{-\V}.\G^L$ \\
	\textrm{Configurations:} &
		$(\bS \vee \bT)^L$ & $\bS^L\oplus \bT^L$ \\
	\textrm{Empty configuration:} & $(\varepsilon)^L$ & $\lzero$ \\
	\\
	\multicolumn{3}{c}{$
		\infer[(\Sct)]
		{
			\exists\bx.\B^L\vdash \exists\bx'.\B'^L
		}
		{
			CT\models \exists\bx.\B\rightarrow\exists\bx'.\B'
		}
		\quad\quad\quad
		\infer[(\Seq)]
		{
			c_u(...,t_j,...) \x !(t_j\doteq u)
			\vdash c_u(...,u,...) \x !(t_j\doteq u)
		}
		{
			\phantom{\top}
		}
	$} \\
	\\
		\multicolumn{3}{c}{$
		\infer[(\Sp)]
		{
			H_1^L\x H_2^L\x G^L \vdash H_1^L \x \exists\by_r.(B^L\x G^L)
		}
		{
			(r\ @\ H_1\setminus H_2\Leftrightarrow G\mid B)\subxy\in\bbP
		}
	$}
	\end{tabular}
}
\end{center}
\caption{The axiomatic linear-logic semantics for CHR$^\vee$}
\label{fig:chrv-axiomatic-semantics}
\end{figure}

\subsubsection{Soundness of the Linear Logic Semantics for CHR$^\vee$}

In this section, we prove the soundness of our semantics with respect to
$\oesqv$. At first, we show that configuration equivalence implies logical
judgement:

\begin{lemma}[$\equv\Rightarrow\dashv\vdash$]
\label{lemma:vee-ce-ll}
\begin{longenum}
 \item \label{prop:vel:goal} For goals $\G_1,\G_2$ such that $\G_1\equiv_G\G_2$,
 we have $\G_1 \dashv\vdash \G_2$.
 \item \label{prop:vel:state} For CHR$^\vee$ states $S_1,S_2$ and an arbitrary
 constraint theory CT such that such that $S_1\equiv S_2$,
 we have $S_1 \dashv\vdash_\Sigma S_2$ where $\Sigma=\Sct$.
 \item \label{prop:vel:conf} For configurations $\bS_1,\bS_2$ and an arbitrary
 constraint theory CT such that such that $\bS_1\equv \bS_2$, we have $S_1
 \dashv\vdash_\Sigma S_2$ where $\Sigma=\Sct$.
\end{longenum}
\begin{proof}
  Lemma~\ref{lemma:vee-ce-ll}.\ref{prop:vel:goal}: The property holds,
  as $\otimes$ is associative, commutative, has the neutral element $1$ and
  distributes over $\oplus$.
  Lemma~\ref{lemma:vee-ce-ll}.\ref{prop:vel:state}: Proof is analogous
  to Lemma~\ref{lemma:sq-ll}.
  Lemma~\ref{lemma:vee-ce-ll}.\ref{prop:vel:conf}: We consider the properties
  given in Def.~\ref{def:vee-config-equiv} --
	Def.~\ref{def:vee-config-equiv}.\ref{cond:vce_ac}: For all $\alpha,\beta,\gamma$,
	we have $\alpha\oplus\beta\dashv\vdash\beta\oplus\alpha$ and
	$(\alpha\oplus\beta)\oplus\gamma\dashv\vdash\alpha\oplus(\beta\oplus\gamma)$.
	Def.~\ref{def:vee-config-equiv}.\ref{cond:vce_steq}: The property follows from Lemma~\ref{lemma:vee-ce-ll}.\ref{prop:vel:state}.
	Def.~\ref{def:vee-config-equiv}.\ref{cond:vce_fail}: For all $\alpha$, we have $0\oplus\alpha\dashv\vdash\alpha$.
	Def.~\ref{def:vee-config-equiv}.\ref{cond:vce_split}: For all $\alpha,\beta,\gamma,\V$, we have $(\exists_{-\V}.\alpha\oplus\beta)\oplus\gamma\dashv\vdash(\exists_{-\V}.\alpha)\oplus(\exists_{-\V}.\beta)\oplus\gamma$.
\end{proof}

\end{lemma}

Theorem~\ref{thm:soundness} states the soundness of the axiomatic linear-logic
semantics for CHR$^\vee$.

\begin{theorem}[Soundness]
  \label{thm:vee-soundness}
  For any CHR$^\vee$ program $\bbP$, constraint theory $CT$ and configurations $\bU,\bV$,
  \[
  	[\bU]\mapsto^{*}[\bV] \quad\Rightarrow\quad \bU^L\vdash_\Sigma \bV^L
  \]
  where $\Sigma=\Sp\cup\Sct$.
\begin{proof}
Let $\bU,\bV$ be configurations such that $\bU\mapsto^{r}\bV$. According to
Def.~\ref{def:chr-op-sem}, there exists a variant of a rule with fresh variables
$(r\ @\ H_1\setminus H_2 \Leftrightarrow G\mid B)$ and configurations
$\bU'=\state{H_1\wedge H_2 \wedge G\wedge\G; \V}\vee\bT'$, $\bV'=\state{B_u
\wedge H_1\wedge B_b\wedge G\wedge\G; \V}\vee\bT'$  such that $\bU'\equiv
\bU$ and $\bV'\equiv \bV$.
Consequently, $\Sp$ contains:
\[
  H_1^L\x H_2^L\x G^L \vdash_\Sigma H_1^L\x \exists\by_r.(B^L\x G^L)
\]
Analogous to the proof of Thm.~\ref{thm:soundness}, we proceed to:
\[
  \exists_{-\V}.H_1^L\x H_2^L\x G^L\x\G
   \vdash_\Sigma
  \exists_{-\V}.H_1^L\x G^L \x B^L\x\G
\]
And then to:
\[
  (\exists_{-\V}.H_1^L\x H_2^L\x G^L\x\G)\oplus {\bT}^L
   \vdash_\Sigma
  (\exists_{-\V}.H_1^L\x G^L \x B^L\x\G)\oplus {\bT}^L
\]
This corresponds to $\bU'^L\vdash_\Sigma \bV'^L$. Lemma~\ref{lemma:sq-ll} then
proves that ${\bU}^L\vdash_\Sigma {\bV}^L$. As the judgement relation
$\vdash_\Sigma$ is transitive and reflexive, the relationship can be generalized
to the reflexive-transitive closure $\bU\mapsto^{*} \bV$.
\end{proof}
\end{theorem}

\subsubsection{Configuration Entailment}

Analogously to state entailment for pure CHR, we define a notion of
\emph{configuration entailment} to characterize the discrepance between
transitions in a CHR$^\vee$ program and judgements in its corresponding sequent
calculus system and thus to completeness of the linear-logic semantics:

\begin{definition}[Entailment of Configurations]
\label{def:c_entail}
Entailment of configurations, denoted as $\cdot\entv\cdot$, is the smallest
reflexive-transitive relation over equivalence classes of configurations satisfying
the following conditions:
  \begin{enumerate}
    \item \label{cond:cn_wea}
    \emph{Weakening:} For any state $S$ and configuration
    $\bar T$:
    \[
        [ \bar T ]
          \entv
        [ S \vee \bar T ]
      \]
    \item \label{cond:cn_stronger}
    \emph{Redundance of Stronger States:} For any CHR$^\vee$
      states $S_1,S_2,T$ such that
      $S_1\ent S_2$:
      \[
        [ S_1 \vee S_2 \vee \bar T ]
          \entv
        [ S_2 \vee \bar T ]
      \]
    \end{enumerate}
\end{definition}

The following property follows from the definition:

\begin{property}[($\ent\Rightarrow\entv$)]
For CHR$^\vee$ states $S_1,S_2$ such that $S_1\ent S_2$:
\[
        [ S_1 \vee \bar T ]
          \entv
        [ S_2 \vee \bar T ]
\]
\begin{proof}
$[ S_1 \vee \bar T ]\entv[ S_2 \vee S_1 \vee \bar T ]=[ S_1 \vee S_2 \vee \bar T
]\entv[ S_2 \vee \bar T ]$
\end{proof}
\end{property}

Lemma~\ref{lemma:vee-exchange} corresponds to Lemma~\ref{lemma:exchange} for
the case of pure CHR.

\begin{lemma}[Exchange of $\mapsto$ and $\entv$]\label{lemma:vee-exchange}
  Let $\bS,\bU,\bT$ be configurations. If $\bS\entv
  \bU$ and $\bU\mapsto^r \bT$ then there exists a configuration $\bV$ such that
  $\bS\mapsto^{*} \bV$ and $\bV\entv \bT$.
\begin{proof}
Firstly, we consider hypothesis with respect to the axioms of configuration
entailment (cf. Def.~\ref{def:c_entail}):
\begin{description}
 \item [Def.~\ref{def:c_entail}.\ref{cond:cn_wea}]
	Assume that $[\bS]\entv[S\vee\bS]\mapsto^{r}[\bT]$. It follows that either (i)
	$[S]\mapsto^r [S']$ and $[\bT]=[S'\vee\bS]$ or (ii) $[\bS]\mapsto^r[\bS']$ and
	$[\bT]=[S\vee\bS']$. In case (i), we have $[\bV]=[\bS]$ and $[\bS] \entv [\bT]$.
	In case (ii), we have $[\bV]=[\bS']$ and $[\bS]\mapsto^r [\bS']\entv
	[S\vee\bS']=[\bT]$.
 \item [Def.~\ref{def:c_entail}.\ref{cond:cn_stronger}]
	Assume that $[S_1\vee S_2\vee\bS]\entv[S_2\vee\bS]\mapsto^{r}[\bT]$ where
	$[S_1]\ent [S_2]$. It follows that either (i) $[S_2]\mapsto^r [S_2']$ and
	$[\bT]=[S_2'\vee\bS]$ or (ii) $[\bS]\mapsto^r[\bS']$ and $[\bT]=[S_2\vee\bS']$.
	In case (i), Lemma~\ref{lemma:exchange} proves that there exists an $S_1'$ such
	that $[S_1]\mapsto^r [S_1']$ and $[S_1']\ent [S_2']$. Hence, we get
	$[\bV]=[S_1'\vee S_2'\vee\bS]$ and $[S_1\vee S_2\vee\bS]\mapsto^r [S_1'\vee
	S_2\vee\bS] \mapsto^r [S_1'\vee S_2'\vee\bS]\entv [S_2'\vee\bS]=[\bT]$. In case
	(ii), we have $[\bV]= [S_1\vee S_2\vee\bS']$ and $[S_1\vee S_2\vee\bS]\mapsto^r
	[S_1\vee S_2\vee\bS'] \entv [S_2\vee\bS]=[\bT]$.
\end{description}
For the reflexive closure of these axioms, the hypothesis is true as
$[\bS]=[\bU]$ implies $[\bV]=[\bT]$. For their transitive closure, it follows by
induction. Hence, the hypothesis holds for configuration entailment in general.
\end{proof}
\end{lemma}

\subsubsection{Completeness of the Linear-Logic Semantics for
CHR$^\vee$}

Lemma~\ref{lemma:vee-proof-structure} sets the stage for the completeness
theorem. Its proof is analogous to the proof of Lemma~\ref{lemma:proof-structure}
and will be omitted here:

\begin{lemma}
\label{lemma:vee-proof-structure}
Let $\pi$ be a cut-reduced proof of a sequent $\bS^L\vdash \bT^L$ where $\bS,\bT$
are arbitrary configurations. Any formula $\alpha$ in $\pi$ is either of the form
$\alpha=\bS_\alpha^L$ or of the form $\alpha=\cb$ where $\bS_\alpha$ is a
configuration and $\cb$ is a built-in constraint.
\end{lemma}

It should be noted that the configuration
$\bS_\alpha$ is not necessarily unique, i.e. more than one configuration might
map to a specific formula. For example, let formula $\alpha=\cu\oplus\cu$. We
then have
$\state{\cu\vee\cu;vars(\cu)}^L
=(\state{\cu;vars(\cu)}\vee\state{\cu;vars(\cu)})^L=\alpha$. However, we
have by Def.~\ref{def:vee-config-equiv}.\ref{cond:vce_split} that
$\bS^L=\bT^L\Rightarrow \bS\entv\bT$.

\begin{theorem}[Completeness of the Semantics for CHR$^\vee$]
\label{thm:vee-completeness}
   Let $\bS,\bT$ be configurations, let $\bbP$ be a program and $CT$ be a
   constraint theory. Then the sequent $\bS^L\vdash \bT^L$ is provable in a
   sequent calculus system with proper axioms $\Sigma=\Sct\cup\Seq\cup\Sp$
   \emph{if and only if} there exists a configuration $\bU$ such that
   $\bS\mapsto^{*} \bU$ and $\bU\ent \bT$.

\begin{proof}
To preserve clarity, we will omit the set of proper axioms from the judgement
symbol. Furthermore, $\cD(\bU,\bV)$ denotes the fact that for configurations
$\bU,\bV$,  there exist configurations $\bU_1,\ldots,\bU_n$ for some $n$ such
that:
	\[
		\bU \mapsto_\vee \bU_1 \mapsto_\vee \ldots \mapsto_\vee \bU_n \entv \bV
	\]
Entailment $\bU\entv \bV$ implies $\cD(\bU,\bV)$. We define
$\cdot\Diamonddot\cdot$ as in the proof of Thm.~\ref{thm:completeness}.

Let $\pi$ be a cut-reduced proof of $\bS^L\vdash \bT^L$.
We assume w.l.o.g. that all existentially quantified variables
in the antecedent of a sequent in $\pi$ are renamed apart. We define $\eta$ as an
extension of the completion function from the proof
of Thm.~\ref{thm:completeness} to configurations:

	\medskip
	\begin{tabular}{l @{\hspace{1mm} $::=$ \hspace{1mm}} l}
	  $\eta(\bS^L)$ & $\bS^L$ \\
	  $\eta(\cb)$ & $!\cb$ \\
	  $\eta(\alpha,\Gamma)$ & $\eta(\alpha)\Diamonddot\eta(\Gamma)$ \\
	  $\eta(\Gamma\vdash\alpha)$ & $\eta(\Gamma)\vdash\eta(\alpha)$
	  	\quad for non-empty $\Gamma$ \\
	  $\eta(\vdash\alpha)$ & $\lone\vdash\eta(\alpha)$ \\
	\end{tabular}

	\medskip
	From Lemma~\ref{lemma:vee-proof-structure} follows that for every
	sequent $\Gamma\vdash\alpha$ in $\pi$, we have
	$\eta(\Gamma\vdash\alpha)=\bU^L\vdash\bV^L$ for some configurations $\bU,\bV$. We
	show by induction over the depth of $\pi$ that for every such $\bU^L\vdash
	\bV^L$, we have $\cD(\bU,\bV)$.

	\textbf{Base case:}
	In case the proof of $\bS^L\vdash \bT^L$ consists in a leaf, it is
	an instance of $(Identity)$, $(L\lzero)$, $(R\lone)$, or a proper axiom
	$(\Gamma\vdash\alpha)\in(\Seq\cup\Sct\cup\Sp)$.
	We apply the same arguments as in the proof of Thm.~\ref{thm:completeness}.
	Thm.~\ref{thm:completeness}.

	\textbf{Induction step:}
	As $\pi$ is cut-reduced, the final inference rule either has to be one of
	$(Cut)$, $(L\otimes)$, $(R\otimes)$, $(L\lone)$, $(Weakening)$, $(Dereliction)$,
	$(Contraction)$, $(R!)$, $(L\exists)$ and $(R\exists)$, or one of $(L\oplus)$,
	$(R\oplus_1)$, and $(R\oplus_2)$. In the former case, we can follow the same
	arguments as in the proof of Thm.~\ref{thm:completeness}. In the following,
	we consider $(L\oplus)$, $(R\oplus_1)$ and $(R\oplus_2)$.

	\begin{itemize}
	  \item $(L\oplus)$:
		\[
		\infer[(L\oplus)]{\Gamma,\alpha\oplus\beta\vdash\gamma}
		    {\Gamma,\alpha\vdash\gamma&\Gamma,\beta\vdash\gamma}
		\]
		Let $\G_\alpha,\G_\beta$ be goals, let
$S_\Gamma=\state{\G;\V}$ be a state
		and let $\bS_\beta$ be a configuration such that
$\G_\alpha^L=\eta(\alpha)$,
		$\G_\beta^L=\eta(\beta)$, $S_\Gamma^L=\eta(\Gamma)$ and
		$\bS_\gamma^L=\eta(\gamma)$. Let furthermore $\by_\alpha=vars(\G_\alpha)$ and $\by_\beta=vars(\G_\beta)$. Hence,
		$\eta(\Gamma,\alpha)=\state{\G\wedge\G_\alpha;\V\cup\by_\alpha}^L$,
		$\eta(\Gamma,\beta)=\state{\G\wedge\G_\beta;\V\cup\by_\beta}^L$, and
		$\eta(\Gamma,\alpha\oplus\beta)=\state{\G\wedge(\G_\alpha \vee
\G_\beta);\V\cup\by_\alpha\cup\by_\beta}^L$.
		The induction hypothesis gives us
		$\cD(\state{\G\wedge\G_\alpha;\V\cup\by_\alpha},\bS_\gamma)$ and
		$\cD(\state{\G\wedge\G_\beta;\V\cup\by_\beta},\bS_\gamma)$.
		By Def.~\ref{def:vee-config-equiv}.\ref{cond:vce_split} we have that
		$\eta(\Gamma,\alpha\oplus\beta)\equiv\state{\G\wedge\G_\alpha;\V\cup\by_\alpha}\vee\state{
		\G\wedge\G_\beta;\V\cup\by_\beta}$.
		Finally by Lemma~\ref{lemma:vee-exchange}, we get $\cD(\state{\G\wedge(\G_\alpha \vee
		\G_\beta);\V\cup\by_\alpha\cup\by_\beta},\bS_\gamma)$.

	  \item $(R\oplus_1)$,  $(R\oplus_2)$:
	  	\[
		\infer[(R\oplus_1)]{\Gamma\vdash\alpha\oplus\beta}
		    {\Gamma\vdash\alpha}
		\quad\quad
		\infer[(R\oplus_2)]{\Gamma\vdash\alpha\oplus\beta}
		    {\Gamma\vdash\beta}
		\]
		We consider $(R\oplus_1)$: By the subformula property, there exist
		configurations $\bS_\Gamma, \bS_\alpha, \bS_\beta$, such that
		$\bS_\Gamma^L=\eta(\Gamma)$, $\bS_\alpha^L=\eta(\alpha)$, and
		$\bS_\beta^L=\eta(\beta)$. By the induction hypothesis, we have
		$\cD(\bS_\Gamma,\bS_\alpha)$. By
		Def.~\ref{def:c_entail}.\ref{cond:cn_wea}, we have
		$\bS_\alpha\entv(\bS_\alpha\vee\bS_\beta)$ and therefore
		$\cD_n(\bS_\Gamma,\bS_\alpha\oplus\bS_\beta)$. (The proof for $(R\oplus_2)$
		works analogously.)
	\end{itemize}

	Finally, we have $\cD(\bS,\bT)$, i.e. there exist configurations $\bS_1,\ldots \bS_n$ such that:
	\[
		\bS \mapsto \bS_1 \ldots \mapsto \bS_n \ent \bT
	\]
	It follows that for $\bU=\bS_n$, we have $\bS\mapsto^{*}\bU$ and $\bU\ent\bT$.
\end{proof}
\end{theorem}

\begin{lemma}[($\entv\Leftrightarrow\,\vdash$)]
\label{lemma:vee-cn-ll}
For configurations $\bS,\bT$, we have $[\bS]\entv[\bT]$ \emph{if and only if} $\bS^L\vdash_\Sigma\bT^L$ where $\Sigma=\Sct\cup\Seq$.
\end{lemma}

\begin{proof}
\noindent ('$\Leftarrow$') Follows from Thm.~\ref{thm:vee-completeness} by
assuming an empty program $\bbP=\emptyset$.

\noindent ('$\Rightarrow$')
We consider the axioms for configuration entailment in
Def.~\ref{def:c_entail}: W.r.t. axiom (\ref{cond:cn_wea}),
$[\bT]\entv[S\vee\bT]$ implies $\bT^L\vdash (S\vee\bT)^L$ since
$\beta\vdash\alpha\oplus\beta$. For
Def.~\ref{def:c_entail}.\ref{cond:cn_stronger}, $[S_1]\ent[S_2]$
implies $S_1^L\vdash_\Sigma S_2^L$ by Lemma~\ref{lemma:sq-ll}. From a proof of
$S_1^L\vdash_\Sigma S_2^L$, we can construct a proof of $S_1^L\oplus
S_2^L\oplus \bT^L \vdash_\Sigma S_2^L\oplus \bT^L$. As $\vdash_{\Sigma}$ is
furthermore reflexive and transitive, the hypothesis is reduced to
Lemma~\ref{lemma:vee-ce-ll}.
\end{proof}

Analogously to the encoding semantics for pure CHR, we define an encoding
semantics for CHR$^\vee$. The translation of states and configurations is
unchanged from the axiomatic semantics. The translation of constraint
theories is the same as in the encoding semantics for pure CHR. The translation
of rules and programs is updated to the syntax of CHR$^\vee$ as shown in
Fig.~\ref{fig:vee-encoding-semantics}.

\begin{figure}
\begin{center}
\fbox{
	\begin{tabular}{l @{\quad} r @{$\,::=\,$} l} \\
	\textrm{Rules:} & $(H_1\setminus H_2\Leftrightarrow G\mid B)^L$
	& $!\forall(H_1^L\x H_2^L\x G^L
	\multimap H_1^L\x \exists\by_r.(B^L\x G^L))$ \\
	\textrm{Programs:} & $\{R_1,\ldots,R_n\}^L$ &
	$\{R_1^L,\ldots,R_n^L\}$ \\
	\\
	\end{tabular}
}
\caption{The linear-logic encoding semantics for CHR$^\vee$}
\end{center}
\label{fig:vee-encoding-semantics}
\end{figure}

The soundness and completeness of the encoding semantics is proven analogously
to Theorem~\ref{thm:vee-completeness}:

\begin{theorem}[Soundness and Completeness of the Encoding Semantics]
  \label{theorem:chrv_embed_soundness_completeness}
  Let $\bS,\bT$ be configurations. There exists a configuration $\bU$ such that
  \[
    \bS\mapsto^{*}\bU\textrm{ and }\bU\ent \bT
  \]
  in a program $\bbP$ and a constraint theory $CT$ \emph{if and only if}
  \[
    \bbP^L,CT^L \vdash \forall(\bS^L\lp \bT^L)
  \]
\end{theorem}

As the encoding semantics is logically equivalent to the one proposed in
\citeN{Betz07}, Theorem~\ref{theorem:chrv_embed_soundness_completeness}
also proves the equivalence of the axiomatic linear-logic semantics with that
earlier semantics.

\subsection{Congruence and Analyticness}
\label{sec:vee-congruence}

The operational semantics $\oesq$ for \emph{pure} CHR features the pleasant
property that state equivalence coincides with mutual entailment of states (cf.
Corollary~\ref{crl:ent-equiv}). In this section, we show that the property of
mutual configuration entailment, henceforth called \emph{congruence of
configurations}, does not in general coincide with configuration equivalence.

To overcome this limitation, we introduce a well-behavedness property on
configurations -- \emph{compactness} -- and one on CHR$^\vee$ programs --
\emph{analyticness} -- which guarantee that congruence coincides with equivalence.

\begin{definition}[Congruence of Configurations] Given a constraint theory $CT$,
two configurations $\bS,\bT$ are considered \emph{congruent} if $\bS\entv\bT$ and
$\bT\entv\bS$. Congruence of $\bS$ and $\bT$ is denoted as $\bS\sim\bT$.
\end{definition}

Congruence of configurations does not generally comply with rule applications
as the following example shows.

\begin{example}[Non-Compliance with Rule Application] By compliance, we mean the
property that for arbitrary configurations $\bS,\bS',\bT$ such that
$\bS\equv\bS'$ and $\bS\mapsto^{*}\bT$, there exists a $\bT'$ such that
$\bS'\mapsto\bT'$ and $\bT'\equv\bT$.

Let $\bS=\state{c_u(X)}$ and $\bT=\state{c_u(0)}\vee\state{c_u(X)}$ be
configurations. As $\state{c_u(0)}\ent\state{c_u(X)}$, we have congruence:
$\bS\sim\bT$. Now consider the following minimal CHR program: \[
 r\ @\ c_u(0)\Leftrightarrow d_u(0)
\] We observe that we have $\bT\mapsto^r\state{d_u(0)}\vee\state{c_u(X)}$
whereas $\bS$ is an answer configuration i.e. it does not allow any further
transition. We thus observe that congruence of configurations is not in general
compliant with rule application.
\end{example}

However, we can make a somewhat weaker statement about the relationship between
congruence and rule application:

\begin{property}[Weak Compliance with Rule Application] Let $\bS,\bS',\bT$ be
configurations such that $\bS\sim\bS'$. Then $\bS\mapsto^{*}\bT$ implies that
there exists a $\bT'$ such that $\bS'\mapsto^{*}\bT'$ and $\bT'\entv\bT$.
\begin{proof}
$\bS\sim\bS'$ implies $\bS'\entv\bS$. Furthermore, we have $\bS\mapsto^{*}\bT$. Hence, Lemma~\ref{lemma:vee-exchange} proves $\bS'\mapsto^{*}\bT'$ and $\bT'\entv\bT$.
\end{proof}
\end{property}

As the congruence relation does not strongly comply with rule application, it is
not appropriate as a general equivalence relation over configurations. On the
other hand, from Lemma~\ref{lemma:vee-cn-ll} follows that congruence of
configurations coincides with logical equivalence over the respective
linear-logic readings:

\begin{property}
For arbitrary configurations $\bS,\bT$, we have $\bS\sim\bT \Leftrightarrow
\bS\equivll\bT$.
\end{property}

Hence, any reasoning over CHR$^\vee$ via the linear-logic semantics is
necessarily modulo congruence. In order to allow precise logical reasoning over
CHR$^\vee$, we identify a segment of CHR$^\vee$ where congruence and equivalence
of configurations coincide. Firstly, we introduce the notion of
\emph{compactness}:

\begin{definition}[Compactness]
 A configuration $\bS$ is called \emph{compact} if it does not have a
 representation $\bS'\equv\bS$ of the form $\bS'=S_1\vee S_2\vee\bS''$ where
 $S_1,S_2$ are flat states such that $S_1\not\equiv S_\bot$ and $S_1\ent S_2$.
\end{definition}

We extend the compactness property to equivalence classes of configurations in
the obvious manner. The following lemma states that compactness guarantees that
congruence and equivalence coincide.

\begin{lemma}
 Let $\bS,\bT$ be compact configurations such that $\bS\sim\bT$. Then $\bS\equv\bT$.
\begin{proof}
 Considering Def.\ref{def:vee-config-equiv}, we observe that every configuration
 $\bS$ has a representation of the form $\bS\equv S_1\vee\ldots\vee S_n$, where
 $S_i=\state{\U_i\wedge\B_i;\V_i}$ for $i\in\{1,\ldots,n\}$. By
 Def.~\ref{def:c_entail}, any two configurations $\bS,\bT$ where $\bS\entv\bT$
 have representations $\bS\equv S_1\vee\ldots\vee S_n, \bT\equv T_1\vee\ldots\vee
 T_m$ such that for every $S_i$ where $S_i\not\equiv S_\bot$, we have the exists
 a $T_j$ such that $S_i\ent T_j)$.

 As $\bS\sim\bT$, we have representations $\bS\equv S_1\vee\ldots\vee S_n,
 \bT\equv T_1\vee\ldots\vee T_m$ such that for every consistent $S_i$, we have a
 $T_j$ such that $S_i\ent T_j$, and for every consistent $T_j$ there is an $S_i$
 such that $T_j\ent S_i$. It follows that for every consistent $S_i$, we have
 $T_j, S_k$ such that $S_i\ent T_j \ent S_k$. As $\bS$ is compact, $S_i\ent S_k$
 implies $S_i\equiv S_k$ and furthermore $S_i\equiv T_j$. As $\bT$ is compact,
 there is exactly one $T_j$ such that $S_i\equiv T_j$.

 Since every consistent $S_i$ has a unique corresponding state $T_j$ with
 $S_i\equiv T_j$ and vice versa, Def.~\ref{def:vee-config-equiv} implies that
 $\bS\equv\bT$.
\end{proof}
\end{lemma}

We furthermore introduce a well-behavedness property for CHR$^\vee$ programs
which guarantees compactness of derived configurations by assuring that
disjoint member states of a configuration have contradicting built-in states. It
appears that a large number of practical CHR$^\vee$ programs satisfy this property.

\begin{definition}[Analytic Program]
\label{def:analytic}
A CHR$^\vee$ program is called \emph{analytic} if for any flat state $S$
and configuration $\bT$ where $[S]\mapsto^{*}[\bT]$, we have that $\bT$ is compact.
\end{definition}

We give a sufficient (although not necessary) criterion for analyticness
of CHR$^\vee$ programs:

\begin{lemma}[Criterion for Analyticness]
\label{lemma:vee-analytic-criterion}

Let $\bbP$ be a CHR$^\vee$ program consisting of rules $R_1,\ldots,R_n$. Assume
that every rule $R_i$ is of the form $r\ @\ H_1\setminus H_2\Leftrightarrow G
\mid (\U_{1}\wedge\B_{1})\vee\ldots\vee(\U_m\wedge\B_m)$ such that
$CT\not\models\exists(B_i\wedge B_j)$ for every $\inintv{i,j}{n}$. Then $\bbP$
is analytic.

\begin{proof} We assume a single rule application $S\mapsto^r\bT$ where the
applied rule be of the form $R_i = r\ @\ H_1\setminus H_2\Leftrightarrow G \mid
(\U_1\wedge\B_1)\vee\ldots\vee(\U_m\wedge\B_m)$ such that
$CT\not\models\exists(B_i\wedge B_j)$ for $\inintv{i,j}{n}$.

It follows that for every
$T_1=\state{\U_1;\B_1;\V_1},T_2=\state{\U_2;\B_2;\V_2}$ such that $\bT\equiv
T_1\vee T_2\vee \bT'$, we have $CT\not\models\exists(\B_1\wedge \B_2)$. It
follows by Lemma~\ref{thm:se_crit} that $T_1\not\ent T_2$.

As the built-in store grows monotonically stronger, correctness for the
transitive closure of $\mapsto$ follows by induction. For the reflexive closure
it follows from the fact that the state $S$ is trivially a compact
configuration.
\end{proof}
\end{lemma}

\section{Application}
\label{sec:application}

In this section, we outline how our results can be applied to reason over
programs and their respective observables. We separate it into two broad
application domains: In Section~\ref{sec:app:observables}, we discuss the
relationship between the linear-logic semantics and program observables. In
Section~\ref{sec:app:comparison}, we show how we can compare the operational
semantics of programs by means of their linear-logic semantics.

\subsection{Reasoning About Observables}
\label{sec:app:observables}

In this section, we show how to apply our results to reason about observables
in both pure CHR and CHR$^\vee$. We will first discuss pure CHR in detail and
then show how the results are generalized to CHR$^\vee$.

\subsubsection{Reasoning About Observables in Pure CHR}

We define two sets of observables based on the linear logic
semantics, paralleling the observable sets of computable states and
data-sufficient answers.

\begin{definition}
Let $\bbP$ be a pure CHR program, $CT$ a constraint theory, and $S$
an initial state. Assuming that $\Sigma=\Sp\cup\Sct\cup\Seq$, we distinguish
two sets of observables based on the linear logic semantics:
\begin{align*}
\cL_{\bbP,CT}^C(S) &
	::= \{[T]\mid \bbP^L,CT^L,S^L\vdash_\Sigma T^L\} \\
\cL_{\bbP,CT}^S(S) &
	::= \{[\state{\top;\B;\V}]\mid \bbP^L,CT^L,S^L\vdash_\Sigma	\state{\top;\B;\V}^L\}
\end{align*}

If the constraint theory $CT$ is clear from
the context or not important, we write the sets as $\cL_{\bbP}^C(S),
\cL_{\bbP}^S(S)$.

\end{definition}

The following definition and property establish the relationship between the
logical observables $\cL_{\bbP}^C$ and $\cL_{\bbP}^S$ and the operational
observables $\cC_{\bbP}$ and $\bcS_{\bbP}$

\begin{definition}[Lower Closure of $\ent$]
For any set $\bbS$ of equivalence classes of CHR states,
\[
	\triangledown\bbS ::= \{[T]|\exists S\in\bbS.[S]\ent[T]\}
\]
\end{definition}

The following property follows directly from
Theorem~\ref{theorem:embed_soundness_completeness}:

\begin{property}[Relationship Between Observables]
For a pure CHR program $\bbP$, a constraint theory $CT$, and an initial
state $S$, we have:
\begin{align*}
	\cL_{\bbP,CT}^C(S) & = \triangledown\cC_{\bbP,CT}(S) \\
	\cL_{\bbP,CT}^S(S) & = \triangledown\cS_{\bbP,CT}(S)
\end{align*}
\end{property}

From this relationship follow several properties that we can use to reason about
the operational semantics. Firstly, in order to prove that a state $S$ cannot
develop into a failed state, it suffices to show that there exists any state $T$,
such that $[T]$ is not contained in $\cC(S)$:

\begin{property}[Exclusion of Failure]
\label{lemma:app-exclude-fail}
Under a program $\bbP$, a constraint theory $CT$, and a CHR state $S$ if there
exists a state $T$ such that $T\not\in\cL_{\bbP,CT}^C(S)$ then
$S_\bot\not\in\cC_{\bbP,CT}(S)$.
\end{property}

Secondly, we can guarantee data-sufficient answers for a state $S$, if we can
prove the empty resource 1 in linear logic:

\begin{property}[Assuring Data-Sufficient Answers]
\label{lemma:app-assure-ds}
\begin{longenum}
\item Under a program $\bbP$, a constraint theory $CT$, and a CHR state $S$, if
$\state{\top;\top;\emptyset}\in\cL_{\bbP,CT}^D(S)$ then $S$ has at least one
data-sufficient answer.
\item If $\bbP$ is furthermore confluent, $S$ has exactly one
data-sufficient answer.
\end{longenum}
\begin{proof}[sketch]
The first property follows from the fact that for any data-sufficient state
$\state{\top;\B;\V}$, we have
$\state{\top;\B;\V}\ent\state{\top;\top;\emptyset}$. The second property
follows from Prop.~\ref{prop:confluence-answers}.
\end{proof}
\end{property}

Finally, if a specific state does not follow in linear logic, it is guaranteed
not to follow in the operational semantics:

\begin{property}[Safety Properties]
\label{lemma:app-safety-properties}
For a program $\bbP$, a constraint theory $CT$, and any two CHR states $S,T$, if
$S'\not\in\cL_{\bbP,CT}^C(S)$ then $S'\not\in\cC_{\bbP,CT}(S)$.
\end{property}

\begin{example}
	\label{example::philosophers}
This example shows how to exploit the completeness of our semantics
to prove safety properties for CHR programs. By \emph{safety property}, we mean
a problem of non-existence of a derivation between two CHR states. The general
form of a safety property is $[T]\not\in\cC_\bbP(S)$.

We implement the $n$-Dining-Philosophers Problem for an arbitrary
number of philosphers and we show using the phase semantics that the program can
never reach a state in which any two philosophers directly neighboring each
other are eating at the same time.

We assume that $CT$
includes the constraint theory for natural numbers.
\[
  \begin{array}{lcl}
    fork(x)\wedge fork(y) & \Leftrightarrow & y=x+1 ~mod~ n \mid eat(x) \\
    eat(x) & \Leftrightarrow & y=x+1 ~mod~ n\mid fork(x)\wedge fork(y) \\
    putfork(0) & \Leftrightarrow & \top \\
    putfork(n) & \Leftrightarrow & n\geq 1 \mid n_1=n-1 \wedge fork(n_1) \wedge
    putfork(n_1)\\
  \end{array}
\]
We want to prove that two philosophers (among $n$ philosophers) which
are seated side by side cannot be eating at the same time. This can be
formalized by the following safety property (we naturally assume
there are at least two philosophers):
\[
  \forall n,i.
  	[\state{eat(i)\wedge eat(j); j=i+1\textrm{ mod }n;\emptyset}]\not\in
  \cL_\bbP^C(
  \state{putfork(n);\top;\emptyset})
\]
Showing that a certain state is not included in $\cL_\bbP^C(S)$, or -- more
generally -- that a certain linear-logic judgement is not valid is in general
not trivial. Having an automated theorem prover try all possible inference
rules exhaustively is an option. In \cite{Betz2008}, a method to
prove safety properties using the phase semantics of linear logic has been
proposed. At this point, it shall suffice to state that we can show:
\[
  \bbP^L, CT^L \not\vdash \exists n,i.(
  putfork(n) \lp eat(i)\x eat(j)\x!(j=i+1)
\]
This proves that no two philosophers seated side by side can be eating at the
same time.
\end{example}

\subsubsection{Generalization to CHR$^\vee$}

As for pure CHR, we define two sets of linear logic observables, paralleling
the sets of computable configurations and data-sufficient answer configurations.

\begin{definition}
Given a CHR$^\vee$ program $\bbP$, a constraint theory $CT$, and an initial
state $S$, we distinguish two sets of observables based on the linear logic
semantics:

\begin{tabular}{r @{\,$::=$} l}
	\multicolumn{2}{c}{}\\
	$\bcL^C_{\bbP,CT}(S)$ &
		$\{[\bT]\mid \bbP^L,CT^L,S^L\vdash \bT^L\}$\\
	$\bcL^S_{\bbP,CT}(S)$ &
		$\{[\state{\top;\B_1;\V_1}\vee\ldots\vee\state{\top;\B_n;\V_n}]\mid$ \\
	\multicolumn{1}{c}{} & \quad
		$\bbP^L,CT^L,S^L\vdash(\state{\top;\B_1;\V_1}\vee\ldots\vee\state{\top;\B_n;\V_n})^L\}$\\
	\multicolumn{2}{c}{}\\
\end{tabular}
\end{definition}

The relationship between the
logical observables and the operational
observables is parallel to pure CHR, though generalized to the
lower closure of configuration entailment

\begin{definition}[Lower Closure of $\entv$]
For any set $\bbS$ of equivalence classes of CHR states,
\[
	\blacktriangledown\bbS ::= \{[\bT]|\exists \bS\in\bbS.[\bS]\entv[\bT]\}
\]
\end{definition}

By Theorem~\ref{theorem:chrv_embed_soundness_completeness}, we then have:

\begin{property}[Relationship Between Observables]
For a CHR$^\vee$ program $\bbP$, a constraint theory $CT$, and an initial
state $S$, we have:
\begin{align*}
	\bcL_{\bbP,CT}^C(S) & = \blacktriangledown\bcC_{\bbP,CT}(S) \\
	\bcL_{\bbP,CT}^S(S) & = \blacktriangledown\bcS_{\bbP,CT}(S)
\end{align*}
\end{property}

Furthermore, each of Property~\ref{lemma:app-exclude-fail},
Property~\ref{lemma:app-assure-ds}, and
Property~\ref{lemma:app-safety-properties} have their obvious
counterparts in CHR$^\vee$.

\subsection{Comparison of Programs}
\label{sec:app:comparison}

In this section, we put special emphasis on the comparison of CHR and CHR$^\vee$
programs across programming paradigms. Hence, we will not treat pure CHR in an
isolated manner but as a subset of CHR$^\vee$. Note also that we use the
encoding rather than the axiomatic formulation of our semantics in this section.

We define three notions of operational equivalence, each one corresponding to
one set of observables as introduced in Section~\ref{sec:op-sem}.

\begin{definition}[Operational Equivalence]
\begin{longenum}
  \item Two CHR$^\vee$ programs $\bbP_1,\bbP_2$ are \emph{operationally
  $\cS$-equivalent} under a given constraint theory $CT$ if for any state $S$, we have
  $\bcS_{\bbP_1,CT}(S)=\bcS_{\bbP_2,CT}(S)$.
  \item Two CHR programs $\bbP_1,\bbP_2$ are \emph{operationally
  $\cA$-equivalent} under a given constraint theory $CT$ if for
  any state $S$, we have $\cA_{\bbP_1,CT}(S)=\cA_{\bbP_2,CT}(S)$.
  \item Two CHR$^\vee$ programs $\bbP_1,\bbP_2$ are \emph{operationally
  $\cC$-equivalent} under a given constraint theory $CT$ if for any state $S$, we have
  $\bcC_{\bbP_1,CT}(S)=\bcC_{\bbP_2,CT}(S)$.
\end{longenum}
\end{definition}

We will mainly focus on $\cC$-eqivalence and $\cS$-equivalence.
What we call $\cA$-equivalence has been researched extensively in the past
(cf. \citeN{DBLP:journals/constraints/AbdennadherFM99}). It shows in this
section that the linear-logic semantics is not adequate to reason about
$\cA$-equivalence.

\begin{definition}[Logical Equivalence of Programs] Two CHR programs
$\bbP_1,\bbP_2$ are called \emph{logically equivalent} under a given constraint
theory $CT$ if $CT^L\vdash\bigotimes\bbP^L_1 \lpl \bigotimes\bbP^L_2$, where the unary
operator $\bigotimes$ stands for element-wise multiplicative conjunction and
$\bigotimes\bbP^L_1 \lpl \bigotimes\bbP^L_2$ is shorthand for
$(\bigotimes\bbP^L_1 \lp \bigotimes\bbP^L_2) \& (\bigotimes\bbP^L_2 \lp
\bigotimes\bbP^L_1)$.
\end{definition}

The following proposition relates $\cC$- and $\cS$-equivalence.

\begin{proposition}\label{prop:op-equiv-hierarchy}
Operational $\cS$-equivalence is a necessary but not a
sufficient condition for $\cC$-equivalence.
\begin{proof}
To show that $\cS$-equivalence is a necessary condition, we assume two
$\cC$-equivalent programs $\bbP_1,\bbP_2$. For every state $S$, we have
$\cC_{\bbP_1}(S)=\cC_{\bbP_2}(S)$. As each $\cS_{\bbP_i}$ is the projection of
$\cC_{\bbP_i}(S)$ to configurations with empty user-defined stores, we also have
$\cS_{\bbP_1}(S)=\cS_{\bbP_2}(S)$.

To show that $\cS$-equivalence is not a sufficient condition, consider the
following two programs:
\begin{align*}
  \bbP_1 = \{ &\quad a(x)\Leftrightarrow b(x) &
  	\bbP_2 = \{ & \quad a(x)\Leftrightarrow x\doteq 0 \\
  & \quad b(x)\Leftrightarrow x\doteq 0 \quad \}&
  & \quad b(x)\Leftrightarrow x\doteq 0 \quad \}&
\end{align*}
Both programs ultimately map every $a(x)$ and $b(x)$ to $x\doteq 0$. Hence, they
are $\cS$-equivalent. For $S=\state{a(x);\emptyset}$ and
$T=\state{b(x);\emptyset}$ we have $[T]\in\bcC_{\bbP_1}(S)$ but
$[T]\not\in\bcC_{\bbP_2}(S)$. Hence, the programs are not $\cC$-equivalent.
\end{proof}
\end{proposition}

We can show that operational $\cC$-equivalence implies logical equivalence of
programs:

\begin{proposition}\label{prop:ll-equiv-impl-confl}Let $\bbP_1, \bbP_2$ be two
$\cC$-equivalent $CHR^\vee$ programs under $CT$. Then
$CT^L\vdash\bigotimes\bbP_1\lpl\bigotimes\bbP_2$.
\begin{proof}
Since $\bbP_1$ and $\bbP_2$ are $\cC$-equivalent, we have that
$\bcC_{\bbP_1}(S)=\bcC_{\bbP_2}(S)$ for all $S$.
For every rule $R=(r\ @\ H_1\setminus
H_2\Leftrightarrow G\mid B)\in\bbP_2$, we have by Def.~\ref{def:vee-tr-system}:
$[\state{H_1\wedge B\wedge
G;\bx}]\in\bcC_{\bbP_2}(\state{H_1\wedge H_2\wedge G;\bx})$
where
$\bx=vars(H_1\wedge H_2\wedge G)$ and then by our hypothesis $[\state{H_1\wedge
B\wedge G;\bx}]\in\bcC_{\bbP_1}(\state{H_1\wedge H_2\wedge G;\bx})$.
Therefore, we get $CT^L\vdash\bigotimes\bbP_1^L\lp R^L$. Applying this to all
rules $R\in\bbP_2$, we show $CT^L\vdash\bigotimes\bbP_1^L\lp\bigotimes\bbP_2^L$.
Analogously, we get $CT^L\vdash\bigotimes\bbP_2^L\lp\bigotimes\bbP_1^L$.
\end{proof}
\end{proposition}

The reverse direction does not hold in general as the following example shows:

\begin{example}
Let the constraint theory $CT$ contain at least the theory of natural numbers.
Compare the following two programs:
\begin{align*}
  \bbP_1 = \{ & \quad c(x)\Leftrightarrow x\geq 1 \quad \}&
  	\bbP_2 = \{ & \quad c(x)\Leftrightarrow \top \\
  & & & \quad c(x)\Leftrightarrow x\geq 1 \quad \}
\end{align*}
The greater-or-equal constraint $\geq$ is a built-in constraint. Hence, it is
translated as $(x\geq 1)^L = !(x\geq 1)$. As $!(x\geq 1)\vdash 1$, we have
$\bigotimes\bbP^L_1\dashv\vdash_\Sigma\bigotimes\bbP^L_2$.
We observe that $\bcS_{\bbP_1}(\state{c(x);x})=\{\state{x\geq 1;x}\}$ and
$\bcS_{\bbP_2}(\state{c(x);x})=\{\state{x\geq 1;x},\state{\top;x}\}$. As the
sets are not equal, $\bbP_1$ and $\bbP_2$ are not operationally $\cS$-equivalent
and hence, by Prop.~\ref{prop:op-equiv-hierarchy}, not $\cC$-equivalent.
\end{example}

However, if we restrict ourselves to analytic, confluent
programs, we can show that logical equivqalence of programs implies operational
$\cS$-equivalence:

\begin{proposition}\label{prop:confl-impl-ll-equiv} Let $\bbP_1, \bbP_2$ be two
analytic confluent $CHR^\vee$ programs such that
$CT^L\vdash\bigotimes\bbP^L_1\lpl\bigotimes\bbP^L_2$.
Then $\bbP_1, \bbP_2$ are $\cS$-equivalent.
\begin{proof}
As both $\bbP_1$ and $\bbP_2$ are confluent, we have
$|\bcS_{\bbP_i,CT}(S)|\in\{0,1\}$ for any state $S$ and $i\in\{1,2\}$,
where $\mid\cdot\mid$ denotes cardinality. If $|\bcS_{\bbP_i,CT}(S)|=0$
then $|\blacktriangledown\bcS_{\bbP_i,CT}(S)|=0$. Otherwise,
$|\blacktriangledown\bcS_{\bbP_i,CT}|\geq 1$. In the former case, our
proposition is trivially true since $\bcS_{\bbP_i,CT}=\emptyset$. In the
following, we assume $|\bcS_{\bbP_i,CT}|=1$.

Logical equivalence implies that $\bcL^C_{\bbP_1,CT}(S)=\bcL^C_{\bbP_2,CT}(S)$
for all $S$. Since $\bcL^S$ is the projection of $\bcL^C$ to configurations with
empty user-defined stores, we also have
$\bcL^S_{\bbP_1,CT}(S)=\bcL^S_{\bbP_2,CT}(S)$ and hence
$\triangledown\bcS_{\bbP_1,CT}(S)=\triangledown\bcS_{\bbP_2,CT}(S)$.

Since $|\bcS_{\bbP_i,CT}(S)|=1$ for $i\in\{1,2\}$, each lower
closure $\triangledown\bcS_{\bbP_i,CT}(S)$ has a maximum
$[\bar M_i]\in\triangledown\bcS_{\bbP_i,CT}(S)$ such that $\forall
[\bS]\in\triangledown\bcS_{\bbP_i,CT}(S).[\bar M_i]\ent [\bS]$ and
$\bcS_{\bbP_i,CT}(S)=\{[M_i]\}$. As
$\blacktriangledown\bcS_{\bbP_1,CT}(S)=\blacktriangledown\bcS_{\bbP_2,CT}(S)$,
we have $\bar M_1\sim\bar M_2$. As both programs are analytic, we furthermore
have that $\bar M_1,\bar M_2$ are compact. Hence, we have $\bar M_1\equv\bar
M_2$ and therefore: $\bcS_{\bbP_1,CT}(S)=\bcS_{\bbP_2,CT}(S)$.
\end{proof}
\end{proposition}

The following example shows that logical equivalence does not imply
operational $\cA$-equivalence:

\begin{example}
We consider the program $\bbP = \{c(x)\Leftrightarrow c(x)\}$ and the empty
program $\bbP_\emptyset=\emptyset$:

As the logical reading $\bbP^L=!\forall(c(x)\lp c(x))$ of $\bbP$ is a logical
tautology, it follows that $\bbP^L\dashv\vdash_\Sigma \bbP_\emptyset^L$ for any
$\Sigma$. Yet, for $S=\state{c(x);\top;\emptyset}$, we have
$\cA_\bbP(S)=\emptyset$ whereas $\cA_{\bbP_\emptyset}(S)=[S]$.
Therefore $\cA_{\bbP}(S)\neq\cA_{\bbP_\emptyset}(S)$. \end{example}

The following final example shows how we can apply the linear-logic
semantics to compare programs across programming paradigms.

\begin{example}
\label{example:append}
We begin with the following
classic Prolog program which implements a ternary \emph{append} predicate for
lists, where the third argument is the concatenation of the first two: \[
  \begin{array}{lcl}
    append(x,y,z) & \leftarrow &
            x\de[~] \wedge y\de z \\
    append(x,y,z) & \leftarrow &
            x\de\left[h | l_1\right] \wedge
            z\de\left[h | l_2\right]
             \wedge ~append(l_1,y,l_2)\\
  \end{array}
\]
We can embed this program into CHR$^\vee$ by explicitly stating the
don't-know non-determinism using the $\vee$ operator.
\begin{align*}
\bbP_1= \{& & append(x,y,z) \Leftrightarrow~ &
    	( {x\de[~]}\wedge {y\de z} ) \vee \\
 & & & (x\de \left[h | l_1\right]\wedge
      z\de\left[h | l_2\right]\wedge append(l_1,y,l_2) )\quad\}
\end{align*}
The linear-logic reading of the embedded program looks as follows:
  \begin{align*}
    \bbP_1^L =
    \{\:
    ! \forall x,y,z.( append(x,y,z)\lp &
          ~\exists l_1,l_2,h\\
	&
       ( !x\de[~] \x !y\de z) \oplus \\
    &
	(
		!x\de[h|l_1] \x !z\de[h|l_2] \x
        append(l_1,y,l_3))
    )
    \:\}
\end{align*}
Secondly, we write a program to implement the \emph{append} predicate the way it
would be expected in CHR:
\begin{align*}
\bbP_2= \{& & append([~],y,z) \Leftrightarrow~ &
    	{y\de z} \\
 & & append([h | l_1],y,z) \Leftrightarrow~ &
 		z\de[h | l_2]\wedge append(l_1,y,l_2) \quad \}
\end{align*}
The two programs are not \emph{per se} $\cS$-equivalent. Consider their
behaviour in case the first argument of $append$ is bound to anything else than a list. For
$S_0=\state{append(3,x,y);\emptyset}$, we have $\bcS_{\bbP_1}(S_0)=\{ S_\bot \}$
but  $\bcS_{\bbP_2}(S_0)=\emptyset$.

Now let us assume that the first argument is always bound to a list. We can model
this by the following formula: \[ \varphi=\forall(append(x,y,z)\lp
append(x,y,z)\otimes (!x\doteq[~] \,\oplus\, \exists h,l.!x\doteq[h|l])) \] It shows that
$CT^L,\varphi\vdash\bigotimes\bbP_1\lpl
\bigotimes\bbP_2$. Hence, under the
assumption that the first argument is always bound to a (non-empty or empty)
list, the two programs are operationally $\cS$-equivalent.

Moreover, we observe that $\varphi$ is equivalent to the logical reading of the
CHR$^\vee$ rule $R_\varphi$: \[ R_\varphi = ( r\ @\ append(x,y,z) \Leftrightarrow
append(x,y,z) \wedge (x\doteq[~] \vee x\doteq[h|l]) ) \] Moreover
$CT^L,\varphi\vdash\bigotimes\bbP_1\lpl\bigotimes\bbP_2$ implies that
$CT^L\vdash(\bigotimes\bbP_1\otimes\varphi)\lpl(\bigotimes\bbP_2\otimes\varphi)$
Hence, the programs $\bbP'_1=\bbP_1\cup R_\varphi$ and $\bbP'_2=\bbP_2\cup
R_\varphi$ are operationally $\cS$-equivalent (without any further assumptions).
\end{example}

\section{Related Work}
\label{sec:related}

From its advent in the 1980ies, linear logic has been studied in relationship
with programming languages.

Common linear logic programming languages such as
LO\cite{DBLP:conf/oopsla/AndreoliP90}, Lolli\cite{DBLP:conf/lics/HodasM91},
LinLog\cite{Andreoli92logicprogramming}, and
Lygon\cite{DBLP:conf/amast/HarlandPW96} rely on generalizations of
backward-chaining backtracking resolution of horn clauses.

The earliest approach at defining a linear-logic semantics for a
committed-choice programming language that we are aware of has been proposed in
\cite{DBLP:conf/kgc/Zlatuska93}. The corresponding language is indeed a fragment
of pure CHR without multiple heads and with substantial restrictions on the use
of built-in constraints.

The linear-logic programming language LolliMon, proposed in
\cite{DBLP:conf/ppdp/LopezPPW05}, integrates backward-chaining proof search with
committed-choice forward reasoning. It is an extension of the aforementioned
language Lolli. The sequent calculus underlying Lolli extended by a set of
dedicated inference rules. The corresponding connectives are syntactically
detached from Lolli's own connectives and operationally they are processed within
a monad. The actual committed-choice behaviour comes by the explicit statement in
the operational semantics, that these inference are to be applied in a
committed-choice manner during proof search. With respect to Lolli, committed
comes thus comes at the cost of giving up the general notion of execution as
proof search, although it is retained outside the monad.

The class LCC of linear logic concurrent constraint programming languages
\cite{DBLP:journals/iandc/FagesRS01} has a close relationship with CHR, although
the former is based on agents whereas the latter is based on rules. Similar to
CHR, LCC languages are non-deterministic and execution is committed-choice. The
linear logic semantics of LCC is similar to our linear logic semantics for pure
CHR and, as far as the two are comparable, it features similar results for
soundness and completeness. Unlike CHR$^\vee$ however, LCC has no notion of
disjunction.

Furthermore, Fages et al. have proposed the so-called
\emph{frontier semantics}\cite{DBLP:journals/iandc/FagesRS01} for LCC, in which
the committed-choice operator is interpreted analogously to the
disjunction operator $\vee$ in CHR$^\vee$. In the linear-logic interpretation of
the frontier semantics, it is correspondingly mapped to the multiplicative
disjunction $\&$. However, the frontier semantics does not constitute a distinct
programming language but is viewed as a tool to reason about properties of LCC
programs. Hence, committed choice never co-exists with disjunction as in the
linear logic semantics for CHR$^\vee$. Rather, the two are viewed as different
interpretations of the same connective for different purposes.

More recently, Simmons et al. proposed the linear logic-based committed-choice
programming language \emph{Linear Logical Algorithms}
\cite{DBLP:conf/icalp/SimmonsP08}. While the language itself corresponds to a
segment of pure CHR, the aim of the work is to define a cost semantics for
algorithms that feature non-deteministic choices.

\section{Conclusion}
\label{sec:conclusion}

In this article, we have presented a detailed analysis of the relationship
between both pure CHR and CHR$^\vee$ with intuitionistic linear logic and we have
shown its applications from reasoning about programs observables to deciding
operational equivalence of multi-paradigm CHR$^\vee$ programs.

Our first main contribution is the linear-logic semantics for the segment of pure
CHR. It encodes both CHR programs and constraint theories to proper axioms of the
sequent calculus. We have shown that equivalence of CHR states coincides with
logical equivalence of the logical readings of state. Furthermore, we have
introduced the notion of state entailment, which precisely characterizes the
discrepance between the transition relation between states in CHR and judgements
between their corresponding logical readings. It is a key notion for the study
and the application of our semantics.

Our second main contribution is the definition of a linear-logic semantics for
CHR$^\vee$. This semantics maps the dualism between don't-care and don't-know
non-determinism in CHR$^\vee$ to the dualism of internal and external choice in
linear logic. Analogously to pure CHR, we have defined a notion of configuration
entailment to characterize the discrepance between state transition and logical
judgement.

We have shown that the linear-logic semantics for CHR$^\vee$ has somewhat less
desirable properties than the one for pure CHR. Concretely, mutual configuration
entailment does not coincide with configuration equivalence. This makes
linear-logic based reasoning over CHR$^\vee$ in general more imprecise. However,
we have presented a well-behavedness property for CHR$^\vee$ -- analyticness --
that amends this limitation.

As our third main contribution, we have shown how to apply our results to reason
about CHR and CHR$^\vee$ programs. We have defined sets of linear-logic based
observables that correspond with the usual program observables of computable
state and data-sufficient answer by means of state entailment or confguration
entailment, respectively. We have presented criteria to prove various program
properties, foremost safety properties, which consist in the non-computability of
a specific state from a certain initial state. Furthermore, we have given a
criterion to prove operational equivalence with respect to data-sufficient
answers for multi-paradigm programs.

As a further contribution, we have for the first time defined an equivalence
relation over configurations and shown its compliance with rule application.
Based on this relation, we have defined an elegant formalization of the
operational semantics of CHR$^\vee$ based on equivalence classes of
configurations. The equivalence-based semantics provides a language to express
properties of programs such as operational equivalence across the boundaries of
programming paradigms.

Our results entail a wide range of possible future work. An obvious line of
future work lies in the application of established methods for automated proof
search in linear logic to reason about CHR and CHR$^\vee$ programs. As
significant effort has been put in the current result on amending the discrepance
between linear judgement and the semantics of CHR, it furthermore suggests itself
to investigate whether a ``purer'' formalism to reason about CHR could be
extracted from linear logic that avoids these discrepances.

\begin{acks} We are grateful to the reviewers of an earlier version of this paper
for their helpful remarks. Hariolf Betz has been funded by the University of Ulm
with LGFG grant \#0518.
\end{acks}

\bibliographystyle{acmtrans}
\bibliography{linlogsem}

\begin{thebibliography}{}

\bibitem[\protect\citeauthoryear{Abdennadher}{Abdennadher}{1997}]{DBLP:conf/cp%
/Abdennadher97}
{\sc Abdennadher, S.} 1997.
\newblock Operational semantics and confluence of constraint propagation rules.
\newblock In {\em CP}, {G.~Smolka}, Ed. Lecture Notes in Computer Science, vol.
  1330. Springer, 252--266.

\bibitem[\protect\citeauthoryear{Abdennadher, Fruhwirth, and Meuss}{Abdennadher
  et~al\mbox{.}}{1996}]{Abdennadher96onconfluence}
{\sc Abdennadher, S.}, {\sc Fruhwirth, T.}, {\sc and} {\sc Meuss, H.} 1996.
\newblock On confluence of {C}onstraint {H}andling {R}ules.
\newblock In {\em CP'96, LNCS 1118}. Springer-Verlag, 1--15.

\bibitem[\protect\citeauthoryear{Abdennadher, Fr{\"u}hwirth, and
  Meuss}{Abdennadher
  et~al\mbox{.}}{1999}]{DBLP:journals/constraints/AbdennadherFM99}
{\sc Abdennadher, S.}, {\sc Fr{\"u}hwirth, T.~W.}, {\sc and} {\sc Meuss, H.}
  1999.
\newblock Confluence and semantics of constraint simplification rules.
\newblock {\em Constraints\/}~{\em 4,\/}~2, 133--165.

\bibitem[\protect\citeauthoryear{Abdennadher and Sch{\"u}tz}{Abdennadher and
  Sch{\"u}tz}{1998}]{DBLP:conf/fqas/AbdennadherS98}
{\sc Abdennadher, S.} {\sc and} {\sc Sch{\"u}tz, H.} 1998.
\newblock {CHR}v: A flexible query language.
\newblock In {\em FQAS}, {T.~Andreasen}, {H.~Christiansen}, {and} {H.~L.
  Larsen}, Eds. Lecture Notes in Computer Science, vol. 1495. Springer, 1--14.

\bibitem[\protect\citeauthoryear{Andreoli}{Andreoli}{1992}]{Andreoli92logicpro%
gramming}
{\sc Andreoli, J.-M.} 1992.
\newblock Logic programming with focusing proofs in linear logic.
\newblock {\em Journal of Logic and Computation\/}~{\em 2}, 297--347.

\bibitem[\protect\citeauthoryear{Andreoli and Pareschi}{Andreoli and
  Pareschi}{1990}]{DBLP:conf/oopsla/AndreoliP90}
{\sc Andreoli, J.-M.} {\sc and} {\sc Pareschi, R.} 1990.
\newblock Lo and behold! {C}oncurrent structured processes.
\newblock In {\em OOPSLA/ECOOP}. 44--56.

\bibitem[\protect\citeauthoryear{Betz}{Betz}{2007}]{Betz07}
{\sc Betz, H.} 2007.
\newblock A linear logic semantics for {C}onstraint {H}andling {R}ules with
  {D}isjunction.
\newblock In {\em Proceedings of the 4th Workshop on Constraint Handling
  Rules}. 17--31.

\bibitem[\protect\citeauthoryear{Betz and Fr{\"u}hwirth}{Betz and
  Fr{\"u}hwirth}{2005}]{DBLP:conf/cp/BetzF05}
{\sc Betz, H.} {\sc and} {\sc Fr{\"u}hwirth, T.~W.} 2005.
\newblock A linear-logic semantics for {C}onstraint {H}andling {R}ules.
\newblock In {\em CP}, {P.~van Beek}, Ed. Lecture Notes in Computer Science,
  vol. 3709. Springer, 137--151.

\bibitem[\protect\citeauthoryear{Betz, Raiser, and Fr{\"u}hwirth}{Betz
  et~al\mbox{.}}{2010}]{DBLP:journals/tplp/BetzRF10}
{\sc Betz, H.}, {\sc Raiser, F.}, {\sc and} {\sc Fr{\"u}hwirth, T.~W.} 2010.
\newblock A complete and terminating execution model for constraint handling
  rules.
\newblock {\em TPLP\/}~{\em 10,\/}~4-6, 597--610.

\bibitem[\protect\citeauthoryear{Duck, Stuckey, de~la Banda, and Holzbaur}{Duck
  et~al\mbox{.}}{2004}]{DBLP:conf/iclp/DuckSBH04}
{\sc Duck, G.~J.}, {\sc Stuckey, P.~J.}, {\sc de~la Banda, M. J.~G.}, {\sc and}
  {\sc Holzbaur, C.} 2004.
\newblock The refined operational semantics of {C}onstraint {H}andling {R}ules.
\newblock In {\em ICLP}, {B.~Demoen} {and} {V.~Lifschitz}, Eds. Lecture Notes
  in Computer Science, vol. 3132. Springer, 90--104.

\bibitem[\protect\citeauthoryear{Fages, Ruet, and Soliman}{Fages
  et~al\mbox{.}}{2001}]{DBLP:journals/iandc/FagesRS01}
{\sc Fages, F.}, {\sc Ruet, P.}, {\sc and} {\sc Soliman, S.} 2001.
\newblock Linear concurrent constraint programming: Operational and phase
  semantics.
\newblock {\em Inf. Comput.\/}~{\em 165,\/}~1, 14--41.

\bibitem[\protect\citeauthoryear{Fr{\"u}hwirth}{Fr{\"u}hwirth}{2009}]{fruehwir%
th09}
{\sc Fr{\"u}hwirth, T.} 2009.
\newblock {\em {C}onstraint {H}andling {R}ules}.
\newblock Cambridge University Press.

\bibitem[\protect\citeauthoryear{Fr{\"u}hwirth and Abdennadher}{Fr{\"u}hwirth
  and Abdennadher}{2003}]{Fruhwirth03}
{\sc Fr{\"u}hwirth, T.} {\sc and} {\sc Abdennadher, S.} 2003.
\newblock {\em Essentials of Constraint Programming}.
\newblock Springer-Verlag New York, Inc., Secaucus, NJ, USA.

\bibitem[\protect\citeauthoryear{Fr{\"u}hwirth}{Fr{\"u}hwirth}{1994}]{DBLP:jou%
rnals/lncs/Fruhwirth94}
{\sc Fr{\"u}hwirth, T.~W.} 1994.
\newblock {C}onstraint {H}andling {R}ules.
\newblock In {\em Constraint Programming}, {A.~Podelski}, Ed. Lecture Notes in
  Computer Science, vol. 910. Springer, 90--107.

\bibitem[\protect\citeauthoryear{Fr{\"u}hwirth}{Fr{\"u}hwirth}{1998}]{DBLP:jou%
rnals/jlp/Fruhwirth98}
{\sc Fr{\"u}hwirth, T.~W.} 1998.
\newblock Theory and practice of {C}onstraint {H}andling {R}ules.
\newblock {\em J. Log. Program.\/}~{\em 37,\/}~1-3, 95--138.

\bibitem[\protect\citeauthoryear{Fr{\"u}hwirth, Pierro, and
  Wiklicky}{Fr{\"u}hwirth
  et~al\mbox{.}}{2002}]{DBLP:journals/entcs/FruhwirthPW02}
{\sc Fr{\"u}hwirth, T.~W.}, {\sc Pierro, A.~D.}, {\sc and} {\sc Wiklicky, H.}
  2002.
\newblock Probabilistic {C}onstraint {H}andling {R}ules.
\newblock {\em Electr. Notes Theor. Comput. Sci.\/}~{\em 76}.

\bibitem[\protect\citeauthoryear{Girard}{Girard}{1987}]{DBLP:journals/tcs/Gira%
rd87}
{\sc Girard, J.-Y.} 1987.
\newblock Linear logic.
\newblock {\em Theor. Comput. Sci.\/}~{\em 50}, 1--102.

\bibitem[\protect\citeauthoryear{Haemmerl\'e and Betz}{Haemmerl\'e and
  Betz}{2008}]{Betz2008}
{\sc Haemmerl\'e, R.} {\sc and} {\sc Betz, H.} 2008.
\newblock Verification of {C}onstraint {H}andling {R}ules using linear logic
  phase semantics.
\newblock In {\em Proceedings of the 5th Workshop on Constraint Handling Rules:
  CHR 2008}.

\bibitem[\protect\citeauthoryear{Harland, Pym, and Winikoff}{Harland
  et~al\mbox{.}}{1996}]{DBLP:conf/amast/HarlandPW96}
{\sc Harland, J.}, {\sc Pym, D.~J.}, {\sc and} {\sc Winikoff, M.} 1996.
\newblock Programming in lygon: An overview.
\newblock In {\em AMAST}, {M.~Wirsing} {and} {M.~Nivat}, Eds. Lecture Notes in
  Computer Science, vol. 1101. Springer, 391--405.

\bibitem[\protect\citeauthoryear{Hodas and Miller}{Hodas and
  Miller}{1991}]{DBLP:conf/lics/HodasM91}
{\sc Hodas, J.~S.} {\sc and} {\sc Miller, D.} 1991.
\newblock Logic programming in a fragment of intuitionistic linear logic.
\newblock In {\em LICS}. IEEE Computer Society, 32--42.

\bibitem[\protect\citeauthoryear{L{\'o}pez, Pfenning, Polakow, and
  Watkins}{L{\'o}pez et~al\mbox{.}}{2005}]{DBLP:conf/ppdp/LopezPPW05}
{\sc L{\'o}pez, P.}, {\sc Pfenning, F.}, {\sc Polakow, J.}, {\sc and} {\sc
  Watkins, K.} 2005.
\newblock Monadic concurrent linear logic programming.
\newblock In {\em PPDP}, {P.~Barahona} {and} {A.~P. Felty}, Eds. ACM, 35--46.

\bibitem[\protect\citeauthoryear{Miller}{Miller}{1992}]{DBLP:conf/elp/Miller92}
{\sc Miller, D.} 1992.
\newblock The pi-calculus as a theory in linear logic: Preliminary results.
\newblock In {\em ELP}, {E.~Lamma} {and} {P.~Mello}, Eds. Lecture Notes in
  Computer Science, vol. 660. Springer, 242--264.

\bibitem[\protect\citeauthoryear{Negri}{Negri}{1995}]{DBLP:journals/mscs/Negri%
95}
{\sc Negri, S.} 1995.
\newblock Semantical observations on the embedding of intuitionistic logic into
  intuitionistic linear logic.
\newblock {\em Mathematical Structures in Computer Science\/}~{\em 5,\/}~1,
  41--68.

\bibitem[\protect\citeauthoryear{Raiser, Betz, and Fr{\"u}hwirth}{Raiser
  et~al\mbox{.}}{2009}]{Raiser2009a}
{\sc Raiser, F.}, {\sc Betz, H.}, {\sc and} {\sc Fr{\"u}hwirth, T.} 2009.
\newblock Equivalence of {CHR} states revisited.
\newblock In {\em 6th International Workshop on Constraint Handling Rules
  (CHR)}, {F.~Raiser} {and} {J.~Sneyers}, Eds. 34--48.

\bibitem[\protect\citeauthoryear{Simmons and Pfenning}{Simmons and
  Pfenning}{2008}]{DBLP:conf/icalp/SimmonsP08}
{\sc Simmons, R.~J.} {\sc and} {\sc Pfenning, F.} 2008.
\newblock Linear logical algorithms.
\newblock In {\em ICALP (2)}, {L.~Aceto}, {I.~Damg{\aa}rd}, {L.~A. Goldberg},
  {M.~M. Halld{\'o}rsson}, {A.~Ing{\'o}lfsd{\'o}ttir}, {and} {I.~Walukiewicz},
  Eds. Lecture Notes in Computer Science, vol. 5126. Springer, 336--347.

\bibitem[\protect\citeauthoryear{Sneyers, Schrijvers, and Demoen}{Sneyers
  et~al\mbox{.}}{2005}]{Sneyers05thecomputational}
{\sc Sneyers, J.}, {\sc Schrijvers, T.}, {\sc and} {\sc Demoen, B.} 2005.
\newblock The computational power and complexity of {C}onstraint {H}andling
  {R}ules.
\newblock In {\em In Second Workshop on Constraint Handling Rules, at ICLP05}.
  3--17.

\bibitem[\protect\citeauthoryear{Zlatuska}{Zlatuska}{1993}]{DBLP:conf/kgc/Zlat%
uska93}
{\sc Zlatuska, J.} 1993.
\newblock Committed-choice concurrent logic programming in linear logic.
\newblock In {\em Kurt G{\"o}del Colloquium}, {G.~Gottlob}, {A.~Leitsch}, {and}
  {D.~Mundici}, Eds. Lecture Notes in Computer Science, vol. 713. Springer,
  337--348.

\end{thebibliography}

\end{document}